\documentclass[11pt,a4paper]{article}
\pdfoutput=1
\usepackage{jheppub}
\usepackage{amsmath,amssymb,mathtools,bm,amsthm}
\usepackage{graphicx,booktabs,comment,appendix}
\usepackage[utf8]{inputenc}
\usepackage{hyperref}
\usepackage{cleveref}
\usepackage{tikz}
\usepackage{pgfplots}
\pgfplotsset{compat=1.18}
\usetikzlibrary{arrows.meta,positioning,calc,decorations.pathmorphing,shapes.geometric}
\hypersetup{colorlinks=true,citecolor=blue,linkcolor=blue,urlcolor=blue}

\newcommand{\Tr}{\operatorname{Tr}}
\newcommand{\cR}{\mathcal R}

\newcommand{\cT}{\mathcal T}

\newcommand{\cO}{\mathcal O}
\newcommand{\dd}{\,\mathrm d}

\newtheorem{theorem}{Theorem}[section]
\newtheorem{lemma}[theorem]{Lemma}
\newtheorem{corollary}[theorem]{Corollary}

\begin{document}

\title{Overcrowding and the Finite-$N$ Hilbert Space}
\author[a,b]{Robert de Mello Koch}
\author[a,b]{Anik Rudra}
\author[c]{Augustine Larweh Mahu}
\affiliation[a]{School of Science, Huzhou Normal University, Huzhou 313000, China}
\affiliation[b]{Mandelstam Institute for Theoretical Physics, School of Physics, University of the Witwatersrand, Private Bag 3, Wits 2050, South Africa}
\affiliation[c]{Department of Mathematics, University of Ghana, LG 62, Legon, Accra, Ghana}
\emailAdd{robert@zjhu.edu.cn}
\emailAdd{anikrudra23@gmail.com}
\emailAdd{almahu@ug.edu.gh}
\date{July 2026}

\abstract{Finite-$N$ trace relations reorganize the Hilbert space of gauge-invariant operators beyond the freely generated large-$N$ description. We study this structure using the Hironaka decomposition of the invariant ring of $d$ Hermitian $N\times N$ matrices. We first prove that the primary invariants may always be chosen to be homogeneous single-trace operators. We then show that, for any such choice, a nontrivial secondary invariant must appear by degree $L_{N,d}=2\log_d N+\log_d\log_d N+\mathcal{O}_d(1)$, which is parametrically below the first universal trace identity at degree $N+1$. This is a global overcrowding effect: exponentially many independent short single traces compete for only $1+(d-1)N^2$ algebraically independent coordinates. The overcrowding scale matches the fastest scrambling times expected for fast scramblers. We argue that this agreement of scales is not accidental: overcrowding provides a microscopic algebraic picture of scrambling in matrix models. Low-rank examples show that secondary invariants can distinguish configurations with identical primary data and, for suitable dynamics, label semiclassical sectors connected by instantons. These results identify the Hironaka decomposition as a natural framework for organizing perturbative and intrinsically finite-$N$ information in collective descriptions of gauge theories.}

\maketitle

\section{Introduction}\label{sec:introduction}

The AdS/CFT correspondence~\cite{Maldacena:1997re,Witten:1998qj,Gubser:1998bc} implies that quantum gravity on a negatively curved spacetime can be equivalently described by a conformal field theory in one lower dimension. This provides a powerful framework for studying non-perturbative aspects of quantum gravity using ordinary quantum field theory. A central element of the AdS/CFT dictionary is that the bulk gravitational coupling is controlled by $1/N^2$. Consequently, finite-$N$ effects in the CFT encode genuinely quantum-gravitational phenomena in the bulk. Understanding the exact finite-$N$ structure of the CFT is an essential step toward understanding quantum gravity beyond the perturbative large-$N$ limit.

A particularly attractive framework for studying finite-$N$ physics is provided by collective field theory~\cite{Jevicki:1979mb,Jevicki:1980zg}. Rather than quantizing the elementary matrix degrees of freedom, collective field theory reformulates the dynamics directly in terms of gauge-invariant variables. It provides one of the few constructive approaches to holography: the emergent bulk description is obtained by quantizing the gauge-invariant operators themselves~\cite{Das:1990kaa,Das:2003vw,deMelloKoch:2023ylr}. This perspective has proved remarkably successful for vector models~\cite{deMelloKoch:2010wdf,deMelloKoch:2018ivk}, where it leads to higher-spin gravity~\cite{Vasiliev:1990en,Vasiliev:1995dn}, and it provides a concrete framework for exploring finite-$N$ holography in matrix~\cite{deMelloKoch:2025ngs,deMelloKoch:2025rkw,deMelloKochKimVanZyl:2025,deMelloKochRodrigues:2026} and vector~\cite{deMelloKoch:2025cec,deMelloKoch:2026dfo,Ahn:2026ohj} models. Since the collective variables are precisely the gauge-invariant operators, finite-$N$ trace relations can be imposed directly within the collective description~\cite{Jevicki:1991yi}. The algebraic structure of the invariant ring is therefore not merely a mathematical curiosity; it determines the fundamental variables of the collective theory and hence the degrees of freedom from which the bulk spacetime is constructed.

At $N=\infty$, the space of gauge-invariant operators is freely generated by the single-trace operators, giving rise to the familiar perturbative Fock-space picture. At finite $N$, however, this description becomes redundant because the single traces satisfy polynomial identities~\cite{Pr,Razmyslov,DrenskyFormanek}. The exact finite-$N$ operator algebra -- and therefore the exact algebra underlying collective field theory -- is naturally described using the Hironaka decomposition~\cite{HochsterRoberts}. In this decomposition the generators are organized into two distinct classes: primary and secondary invariants. The primary invariants $P_a$ are algebraically independent and generate a polynomial subalgebra. For a matrix quantum mechanics with $d\geq2$ species of Hermitian matrices, the number of primary invariants is
\begin{equation}
 h=1+(d-1)N^2.\nonumber
\end{equation}
The secondary invariants $\eta_\alpha$, in contrast, form a finite module over this polynomial algebra. Every gauge-invariant operator admits a unique decomposition
\begin{equation}
 \mathcal O=\sum_{\alpha=1}^{N_s}f_\alpha(P_a)\eta_\alpha, \label{eq:intro-hironaka}
\end{equation}
where each $f_\alpha(P_a)$ is a polynomial in the primary invariants. For natural choices of primary invariants, the number of secondary invariants grows as $N_s\sim e^{cN^2}$, with $c=O(1)$.\footnote{The number of secondary invariants depends on the specific choice made for the primary invariants, so it is not an intrinsic invariant of the ring. For an example, see Section \ref{subsec:choice-dependence}. Reference~\cite{deMelloKochKimVanZyl:2025} argued for $e^{O(N^2)}$ growth for a natural choice of primary invariants. The value of $c$ in $N_s\sim e^{cN^2}$ is determined only after a particular set of primary invariants has been chosen.}

The relation between the invariant ring and the physical Hilbert space is especially direct in the gauged matrix harmonic oscillator. Gauge-invariant states are obtained by acting on the Fock vacuum with invariant polynomials in the adjoint creation operators. The map
\begin{equation}
 f\in{\cal R}_{N,d}\quad\longmapsto\quad f(A_1^\dagger,\ldots,A_d^\dagger)|0\rangle
 \label{eq:intro-ring-hilbert}
\end{equation}
is an isomorphism of graded vector spaces between the invariant ring and the singlet Fock space. Consequently, the Hironaka decomposition induces an exact decomposition of the singlet Hilbert space into a finite collection of primary towers. The primary invariants behave as perturbative degrees of freedom: arbitrary polynomial excitations along a primary direction are allowed. The secondary invariants do not generate new free oscillator directions, since products of secondaries reduce to the finite secondary basis with coefficients depending on the primaries -- see equation \eqref{eq:secondary-multiplication} below.. They are therefore more naturally interpreted as intrinsically finite-$N$ data that fiber over the perturbative primary-coordinate space. Previous work provides two suggestive indications that this finite data can have a non-perturbative interpretation: its size can be of order $e^{cN^2}$~\cite{deMelloKochKimVanZyl:2025} reproducing the characteristic entropy scale associated with black-hole degrees of freedom, and matrix integrals rewritten in invariant variables involve an integral over the primary data together with a sum over branches distinguished by the secondary data~\cite{deMelloKochRodrigues:2026}.

Within collective field theory the primary invariants play a distinguished role: they provide independent coordinates on the space of gauge-invariant configurations. If the interpretation of the primaries as perturbative degrees of freedom is correct, holography suggests that they should admit a basis consisting entirely of single-trace operators, since single traces are dual to perturbative single-particle excitations in the bulk. The first result of this paper establishes precisely this statement. We prove that a homogeneous system of parameters can always be chosen to consist entirely of single traces. This provides a natural bridge between the algebraic organization of the invariant ring and the standard holographic dictionary.

This result has a striking consequence. The first trace relations appear only at degree $N+1$, so below this degree the single-trace operators remain linearly independent in the stable range. Their number nevertheless grows exponentially with their length, whereas only $1+(d-1)N^2$ directions can serve as algebraically independent primary coordinates. We show that by the degree
\[
 L_{N,d}=2\log_dN+\log_d\log_dN+O_d(1)
\]
at least one single-trace direction must be omitted from any all-single-trace primary algebra and hence must appear in the secondary module. We refer to this global capacity constraint as \emph{overcrowding}. It occurs parametrically before the first trace identity becomes visible: the relevant traces remain linearly independent, but they cannot all be promoted simultaneously to freely generated coordinates of the complete finite-$N$ invariant algebra.

The logarithmic dependence of the overcrowding scale is the same as that expected for the scrambling time of a fast scrambler~\cite{Sekino:2008he}. The connection becomes concrete if one assumes that the word length of a typical operator grows linearly with time and that chaotic evolution explores the available single-trace directions sufficiently democratically. Under these assumptions, the counting argument that determines when the exponentially proliferating single-trace operators exhaust the $O(N^2)$ perturbative directions is the same as the argument determining when an initially localized perturbation has spread over an $O(1)$ fraction of those directions. We therefore propose that overcrowding supplies a microscopic algebraic explanation of scrambling in matrix models.

A deeper understanding of overcrowding also requires a clearer physical interpretation of the secondary sector. We therefore study simple low-rank examples in which secondary invariants distinguish different sheets of the invariant configuration space that share identical primary data. For suitable dynamics, transitions between these sheets are mediated by instantons and are exponentially suppressed, rendering the mixing between sectors invisible in perturbation theory about any one semiclassical vacuum. These examples demonstrate that the algebraic organization can acquire a non-perturbative dynamical meaning, although they do not imply that every secondary is a universal non-perturbative sector for every Hamiltonian.

The paper is organized as follows. Sections \ref{sec:background}--\ref{sec:first-secondary} develop the invariant-theory results and derive the overcrowding bound. Section~\ref{fastscrambling} discusses its possible connection with fast scrambling, while Sections \ref{sec:z2}--\ref{sec:four-matrix} use a toy model and low-rank matrix models to investigate the physical meaning of the secondary sector.

Given the overlap in subject matter with~\cite{deMelloKoch:2025ngs,deMelloKoch:2025rkw,deMelloKochKimVanZyl:2025,deMelloKochRodrigues:2026}, it is useful to summarize the new results established here. The principal contributions of this paper are the construction of an all-single-trace homogeneous system of parameters, the identification in the stable range of the first omitted single-trace direction with the first nontrivial secondary invariant, the derivation of an $O(\log N)$ overcrowding bound, and a proposed connection of this structure to fast scrambling and to instanton-mediated orientation and Gram-sheet sectors.

\section{Primary and Secondary Invariants}\label{sec:background}

Consider a matrix model of $d$ Hermitian $N\times N$ matrices,
\begin{equation}
        X^a=X^{a\,\dagger},  \qquad a=1,\ldots,d.
\end{equation}
The model has a $U(N)$ gauge symmetry and the $X^a$ transform in the adjoint representation
\begin{equation}
        X^a\longmapsto U X^a U^\dagger,   \qquad U\in U(N).   \label{eq:unitary-conjugation}
\end{equation}
The configuration space of the $X^a$ is a real vector space ${\rm Herm}_N^d$ of dimension $dN^2$. Define the ring $\mathcal R_{N,d}$ to be the algebra of complex-valued polynomial functions on this space, that are invariant under the $U(N)$ gauge symmetry,
\begin{equation}
        \mathcal R_{N,d} =\mathbb C[{\rm Herm}_N^d]^{U(N)}.\label{eq:hermitian-invariant-ring}
\end{equation}
The physical gauge-invariant configuration space is the orbit space ${\rm Herm}_N^d/U(N)$; $\mathcal R_{N,d}$ is its algebra of polynomial observables.

For our invariant-theory arguments it is useful to complexify this real representation. After complexifying we obtain the space of complex matrices $\mathrm{Mat}_N(\mathbb C)$. The complexified version of $U(N)$ conjugation symmetry is $GL_N(\mathbb C)$ conjugation symmetry. Since imposing Hermiticity does not change the algebraic problem of finding polynomial gauge-invariant operators\footnote{It only selects a real form of the complexified invariant theory.}, studying polynomial invariants our Hermitian problem is equivalent to studying the polynomial invariants of $d$ complex matrices acted on by simultaneous $GL_N(\mathbb C)$ conjugation
\begin{equation}
\mathcal R_{N,d}\simeq\mathbb C[\mathrm{Mat}_N(\mathbb C)^d]^{GL_N(\mathbb C)}.
\label{eq:invariant-ring}
\end{equation}
The matrices are Hermitian in the physical model, while the complexified description supplies the algebraic-geometric tools used below.

The invariant ring is generated by traces of words \cite{Pr,Razmyslov},
\begin{equation}
        t_w=\Tr\left(X^{i_1}X^{i_2}\cdots X^{i_m}\right), \qquad w=i_1i_2\cdots i_m. \label{eq:trace-word}
\end{equation}
The degree of $t_w$ is the word length $m$.  Cyclically related words give the same trace. At $N=\infty$, distinct cyclic words are independent generators and in any fixed total degree, the large-$N$ algebra is a polynomial algebra freely generated by the single traces.  At finite $N$ this description is modified in two distinct ways.  First, trace relations relate operators that are independent at infinite $N$.  The first trace relation has degree $N+1$ \cite{Pr,Razmyslov,DrenskyFormanek}.  Second, and more important for the present work, the finite-$N$ invariant ring has only a finite number of algebraically independent generators.  This second effect will have consequences even at degrees far below the first trace identity. For $d=1$ there are $N$ algebraically independent traces which can be taken to be $\Tr(X^n)$ for $n=1,2,\cdots,N$ and polynomials in these traces generate all gauge invariant operators. Equivalently for $d=1$ we could have taken the eigenvalues as the algebraically independent generators. The conclusion is that for $d=1$ there are $N$ algebraically independent generators. Now consider $d\geq2$. The real dimension of $\mathrm{Herm}_N^d$ is $dN^2$.  A generic tuple has only the central $U(1)$ as its stabilizer, so a generic $U(N)$ orbit has real dimension $N^2-1$.  The orbit space thus has real dimension
\begin{equation}
        h=dN^2-(N^2-1)=1+(d-1)N^2. \label{eq:number-of-primaries}
\end{equation}
Simply put, after a complete gauge fixing $h=1+(d-1)N^2$ degrees of freedom remain. The conclusion is that for $d\ge 2$ there are $1+(d-1)N^2$ algebraically independent generators. This explains the number of primaries quoted in the introduction. In the terminology of ring theory, $h$ is the Krull dimension of the complexified invariant ring in \eqref{eq:invariant-ring}. So, a generic gauge orbit is specified by $h$ independent continuous invariants, even though the number of distinct single traces grows exponentially with their length.

\subsection{The Hironaka decomposition}\label{subsec:hironaka}

At $N=\infty$, the invariant ring is freely generated by the single-trace invariants. At finite $N$, one can still choose a finite set of generators, but the structure of the ring is richer: the generators naturally split into primary and secondary invariants. There are $h$ primary invariants in total, they are algebraically independent and act freely. The complete set of primary invariants is referred to as a \emph{homogeneous system of parameters} (hsop). The invariant ring $\mathcal R_{N,d}$ is a finite module over the polynomial ring generated by the hsop
\begin{equation}
        \mathcal P=\mathbb C[p_1,\ldots,p_h].
\end{equation}
Algebraic independence says that the primary invariants are genuine coordinates on a dense open subset of the invariant configuration space.  Fixing the primary coordinates leaves only a finite set of algebraic information undetermined.

The gauge group $GL_N(\mathbb C)$ is reductive. The Hochster--Roberts theorem tells us that the invariant ring \eqref{eq:invariant-ring} is Cohen--Macaulay \cite{HochsterRoberts}. The important consequence of this for our discussion is that the finite module over $\mathcal P$ is then certain to be free. There are homogeneous invariants $\eta_\alpha$, with $\eta_0=1$, such that
\begin{equation}
        \mathcal R_{N,d}=\bigoplus_{\alpha=0}^{N_s-1}\eta_\alpha\,\mathcal P. \label{eq:hironaka-decomposition}
\end{equation}
The $\eta_\alpha$ are the \emph{secondary invariants}.  Every gauge-invariant operator has a unique expansion
\begin{equation}
        \mathcal O=\sum_{\alpha=0}^{N_s-1}\eta_\alpha F_\alpha(p_1,\ldots,p_h),\qquad F_\alpha\in\mathcal P.
        \label{eq:unique-hironaka-expansion}
\end{equation}
The distinction between primaries and secondaries is a distinction between the polynomial coordinate algebra and a finite basis for the full ring as a module over that algebra. For a homogeneous decomposition, the Hilbert series is
\begin{equation}
  H_{N,d}(q)=\frac{\sum_{\alpha=0}^{N_s-1} q^{\deg\eta_\alpha}}{\prod_{a=1}^{h}\left(1-q^{\deg p_a}\right)}.
        \label{eq:hironaka-hilbert-series}
\end{equation}
The denominator records the freely generated primary directions, while the numerator records the finite secondary module.  The secondary invariants obey the following multiplication rule
\begin{equation}
        \eta_\alpha\eta_\beta=\sum_{\gamma=0}^{N_s-1}C_{\alpha\beta}^{\ \ \gamma}(p)\,\eta_\gamma,
        \qquad C_{\alpha\beta}^{\ \ \gamma}(p)\in\mathcal P. \label{eq:secondary-multiplication}
\end{equation}
This finite multiplication table is precisely the information that is lost if the invariant algebra is approximated by a freely generated Fock space, i.e. if the trace relations are ignored.

In summary, the finite-N trace relations are encoded by the Hironaka decomposition \eqref{eq:hironaka-decomposition}: the primaries act freely, while the secondaries appear linearly and close under the product rule \eqref{eq:secondary-multiplication}.

\subsection{Physics interpretation}
\label{subsec:physical-interpretation}

To develop a physical interpretation of the Hironaka decomposition, we will first outline a precise correspondence between the invariant ring and the gauge-invariant Fock space. Towards this end, consider the gauged $d$-matrix harmonic oscillator. We use the matrices $X^a$ and their conjugate momenta to form the bosonic creation operators $A_a^\dagger$, $a=1,\ldots,d$, which transform in the adjoint representation of the gauge group,
\begin{equation}
A_a^\dagger\longmapsto U A_a^\dagger U^\dagger,\qquad U\in U(N).
\end{equation}
Let $\mathcal R_{N,d}$ denote the ring of polynomial invariants constructed from the $A_a^\dagger$. Every polynomial $f\in\mathcal R_{N,d}$ defines a singlet state
\begin{equation}
\Phi(f)=f(A_1^\dagger,\ldots,A_d^\dagger)|0\rangle .
\end{equation}
Conversely, every state in the singlet Fock space can be obtained in this way. Moreover, the map from invariants to Fock space states is injective: distinct normally ordered monomials in the bosonic creation operators produce linearly independent Fock states, so that
\begin{equation}
f(A_1^\dagger,\ldots,A_d^\dagger)|0\rangle=0\qquad\Rightarrow\qquad f=0 .
\end{equation}
It follows that $\mathcal H_{\rm singlet}\simeq\mathcal R_{N,d}$ as graded vector spaces, with polynomial degree mapped to oscillator excitation number. The Hironaka decomposition
\begin{equation}
\mathcal R_{N,d}=\bigoplus_{\alpha=0}^{N_s-1}\eta_\alpha\mathcal P\qquad \mathcal P=\mathbb C[p_1,\ldots,p_h],
\end{equation}
therefore induces the exact decomposition of Fock space
\begin{equation}
\mathcal H_{\rm singlet}=\bigoplus_{\alpha=0}^{N_s-1}\mathcal P(A^\dagger)\eta_\alpha(A^\dagger)|0\rangle .\label{FockDecomp}
\end{equation}
In this way the Hironaka decomposition of the invariant ring is also a decomposition of the singlet Fock space. There is a natural interpretation of the decomposition \eqref{FockDecomp}.  The primary invariants act freely and consequently mimic Fock oscillators: each primary invariant can be raised to any power, i.e. it can be excited an arbitrary number of times. Each secondary instead supplies a new seed for a complete primary tower,
\begin{equation}
        \eta_\alpha,  \quad \eta_\alpha p_a, \quad \eta_\alpha p_ap_b, \quad\ldots. \label{eq:secondary-tower}
\end{equation}
This motivates the interpretation of the secondary basis as a finite set of distinguished sectors or non-perturbative states, with ordinary perturbative excitations built on top of each sector~\cite{deMelloKoch:2025ngs,deMelloKochRodrigues:2026}.  This interpretation \emph{is not proved}, but it is also not merely terminological as we now explain. In simple matrix integrals, we can express the integral as an integral over the invariants~\cite{deMelloKochRodrigues:2026}. The result is suggestive: primary invariants appear as integration variables. In the path integral we integrate over perturbative degrees of freedom. The complete integral is a sum over branches over the primary space with the number of branches equal to the number of secondary invariants. Different secondary values distinguish branches of invariant configuration space. In the path integral we sum over non-perturbative saddles.

There is also a suggestive entropy count.  Previous work~\cite{deMelloKochKimVanZyl:2025} comparing invariant counting with the restricted-Schur polynomial description of the invariants\footnote{The restricted Schur polynomials provide a vector space description of the invariants~\cite{Bhattacharyya:2008rb,Bhattacharyya:2008xy,deMelloKoch:2012sie}. By summing linear combinations of restricted Schur polynomials we can reconstruct \emph{any} gauge invariant operator. The invariant theory description employed here is algebraic; by multiplying invariants and summing as in \eqref{eq:unique-hironaka-expansion} we can generate \emph{any} gauge invariant operator. For other bases of gauge invariant operators, equivalent to the restricted Schur polynomals see~\cite{Kimura:2007wy,Brown:2007xh,Brown:2008ij}.} argued that, at fixed $d$ and large $N$, the number of secondary invariants behaves as
\begin{equation}
        N_s\sim \exp\left(c N^2\right),   \qquad c=O(1).  \label{eq:secondary-growth}
\end{equation}
Thus the finite module carries an entropy of $\log N_s=O(N^2)$~\cite{deMelloKochKimVanZyl:2025} which matches a black-hole entropy.  The agreement is structural: the Cohen--Macaulay theorem does not, by itself, identify an individual secondary with a black-hole microstate.  It does show that the exact finite-$N$ operator algebra contains an exponentially large, intrinsically non-free sector with precisely the parametric size expected of non-perturbative gravitational states.

The above decomposition of the graded vector space of singlet states is not, in general, an orthogonal decomposition with respect to the Fock-space inner product. Further, the individual towers need not be preserved by an interacting Hamiltonian. These are additional dynamical questions. Nevertheless, the decomposition is exact at the level of the finite-$N$ state space and it supplies the algebraic structure whose physical consequences we study below. More specifically, the central question of this paper is when the additional ``secondary'' structure is forced to appear.  Trace identities are unavoidable at degree $N+1$. We argue that the first secondary invariant is forced much earlier, at a degree of order $\log N$.  The mechanism is an ``overcrowding'' phenomenon: exponentially many short single traces must fit into only $1+(d-1)N^2$ independent primary-coordinate slots.

\section{Single-Trace Primary Invariants}\label{sec:single-trace-primaries}

The Hironaka decomposition is not unique.  The goal of this section is to argue that it is always possible to choose the independent primary coordinates to be single traces. This result is physically appealing: single traces are holographically dual to perturbative single particle excitations in the bulk, which supports the interpretation of primary invariants as perturbative degrees of freedom.

For each degree $m$, define
\begin{equation}
 \cT_m=\operatorname{span}_{\mathbb C} \left\{\Tr(X^{i_1}\cdots X^{i_m})\right\} \subset (\cR_{N,d})_m.
 \label{eq:single-trace-space}
\end{equation}
An element of $\cT_m$ is a complex linear combination of single trace operators, with each single trace operator of degree $m$.

\begin{theorem}[Single-trace primary coordinates]\label{thm:single-trace-hsop}
For every $N\geq1$ and $d\geq2$, the invariant ring $\cR_{N,d}$ admits a homogeneous system of parameters
\begin{equation}
 p_1,\ldots,p_h, \qquad h=1+(d-1)N^2,
 \end{equation}
with
\begin{equation}
 p_a\in\cT_{m_a},  \qquad p_a=\Tr P_a(X^1,\ldots,X^d).  \label{eq:all-single-trace-hsop}
\end{equation}
Thus no multi-trace operator is forced into the independent primary coordinate set.
\end{theorem}

The proof has a simple physical core. A continuous direction in the space of gauge-invariant data is a continuous family of matrix configurations, modulo gauge transformations, along which the gauge-invariant data genuinely changes. The primary invariants are the coordinates that detect these continuous directions: once the primaries are fixed, no continuous gauge-invariant freedom remains.

To argue that the primary coordinates can be chosen to be single-trace operators, note that all polynomial gauge-invariant observables are generated by traces. Consequently, if every single trace is constant along some motion, then every gauge-invariant polynomial is constant along that motion. Such a motion is invisible to all gauge-invariant measurements, and hence is not a genuine direction in the invariant space. \emph{It follows that the single traces already detect all genuine continuous directions.} We may therefore choose an algebraically independent set of single traces as the primary invariants, while the secondary invariants describe the remaining finite algebraic data over these coordinates.

\subsection{An operational criterion}\label{subsec:operational-hsop-criterion}

Let $\mathcal M_{N,d}$ denote the space of consistent gauge-invariant data. A point of $\mathcal M_{N,d}$ is specified by assigning values to all polynomial gauge-invariant operators, subject to the finite-$N$ trace relations. $\mathcal M_{N,d}$ is the space obtained after replacing matrix configurations by the values of their gauge-invariant observables\footnote{Equivalently, in algebraic-geometric language, $\mathcal M_{N,d}$ is the affine variety $\mathcal M_{N,d}={\rm Spec}\,\mathcal R_{N,d}$. We will not need this terminology in what follows; the important point is simply that $\mathcal M_{N,d}$ parametrizes the possible values of the gauge-invariant polynomial observables.}.

There is a distinguished point in this space, which we call the \emph{invariant origin} and denote by $o$. It is the point at which every positive-degree invariant vanishes. On the physical Hermitian slice, this point is represented only by the zero configuration. Indeed, if $X^a=X^{a\dagger}$, then $\Tr(X^a X^a)=0$ implies $X^a=0$. In the complexified configuration space, however, the same invariant data can arise from nonzero nilpotent configurations: these lie in the null cone and have all positive-degree invariants equal to zero. The proof below is naturally formulated in this complexified setting.

We will use the following criterion.  Let $p_1,\ldots,p_h$ be homogeneous invariants.  If their common zero set in $\mathcal M_{N,d}$ consists only of $o$, then
\begin{equation}
        \mathcal R_{N,d}/(p_1,\ldots,p_h) \label{eq:finite-quotient}
\end{equation}
is finite-dimensional, and the $p_a$ form a homogeneous system of parameters. The algebraic proof is given in
Appendix~\ref{app:hsop-criterion}.  Intuitively, the equations $p_a=0$ have removed every continuous invariant direction.  Since there are exactly $h=\dim\mathcal M_{N,d}$ of them, they provide a complete set of independent
coordinates.

The theorem we are trying to establish is therefore reduced to the following constructive question: Can one choose $h$ homogeneous single traces whose simultaneous vanishing leaves only the invariant origin?  The answer is yes.

\subsection{Cutting down all invariant branches} \label{subsec:cutting-branches}

Assume that homogeneous single traces $p_1,\ldots,p_r$ have already been chosen, and consider their common zero set
\begin{equation}
        Z_r=\{p_1=\cdots=p_r=0\}        \subset\mathcal M_{N,d}.        \label{eq:remaining-zero-set}
\end{equation}
Write the positive-dimensional irreducible components of $Z_r$ as $C_1,\ldots,C_s$.  These components are the continuous branches of invariant data remaining after the first $r$ conditions are imposed. Choose a non-origin point $q_j\in C_j$ on each branch.  Because traces of words generate the invariant ring, some single trace is nonzero at $q_j$.  The technical issue is to choose the detecting traces homogeneously and to make one choice work on all surviving branches at once. The detecting traces for different branches need not initially have the same degree. The following lemma removes this mismatch.
\begin{lemma}[Common degree]\label{lem:common-degree}
Given finitely many non-origin points $q_1,\ldots,q_s$ of $\mathcal M_{N,d}$, there is an integer $M\geq1$ such that, for every $j$, some element $t_j\in\mathcal T_M$ obeys
\begin{equation}
        t_j(q_j)\neq0.
\end{equation}
\end{lemma}
A proof is given in appendix~\ref{app:common-degree}.  The idea is to start with one nonzero word trace at each point, repeat the words until their lengths agree, and then take a further common power.  The latter power is chosen so that the phases of all eigenvalues of maximal modulus align; the dominant contributions to the traces then add rather than cancel.

Fix the degree $M$ supplied by the lemma.  For each branch $C_j$, define
\begin{equation}
        K_j=\{t\in\mathcal T_M:\ t|_{C_j}=0\}.\label{eq:bad-linear-subspace}
\end{equation}
This is a proper linear subspace of $\mathcal T_M$, because the element $t_j$ is nonzero at the chosen point $q_j\in C_j$. A finite union of proper linear subspaces cannot fill a complex vector space. We may therefore choose
\begin{equation}
        p_{r+1}\in \mathcal T_M\setminus\bigcup_{j=1}^{s}K_j. \label{eq:generic-single-trace-cut}
\end{equation}
The new single trace is not identically zero on any surviving positive-dimensional branch.  Its vanishing consequently imposes one genuine polynomial condition on every such branch and lowers its dimension by one.

Starting from a space of dimension $h$, we repeat this construction $h$ times.  After the $h$th step no positive-dimensional component remains. The common zero set $Z_h$ is homogeneous because all the $p_a$ are homogeneous. If $Z_h$ contained a nonzero point, rescaling all matrices would generate a one-parameter family of nonzero solutions, contradicting its being zero-dimensional.  The only remaining invariant data are therefore those of the origin and so $Z_h=\{o\}$.  The criterion of subsection~\ref{subsec:operational-hsop-criterion} now proves Theorem~\ref{thm:single-trace-hsop}.

We end this section with two comments.  First, the theorem is an existence statement; it does not select a canonical set of primaries.  Second, a ``single trace'' in the theorem means any complex linear combination of single trace operators, not the trace of a single word.  A generic linear combination is needed in general, to produce a single observable capable of detecting several branches simultaneously.  Is it possible that every primary can be chosen as the trace of a single monomial word? Explicit computations seem to suggest this is not the case. For example, the possibility of such a stronger choice can be considered in the example for which $N=2$ and $d=4$. In this case there are 13 primary invariants. We can search for a system of primary invariants such that there are 4 degree one invariants and 9 degree two invariants. An explicit analysis~\cite{deMelloKoch:2025ngs} shows that, with this degree assignment, no homogeneous system of parameters can be chosen in which every primary invariant is the trace of a single monomial word. The point is that generic linear combinations such as $\Tr(X^1X^2)-\Tr(X^3X^4)$ are needed for the natural low-degree hsop.

\section{The First Secondary}\label{sec:first-secondary}

This section derives the key algebraic result of this paper. Using the result derived in the last section, together with well known facts about the growth of the number of single trace operators, we can now ask at what degree, denoted $L_{N,d}$, the first secondary invariant must appear in the invariant algebra of $d$ $N\times N$ matrices. We find that $L_{N,d}=2\log_d N+\log_d\log_d N+O_d(1)$. This result is an upper bound: if we are systematically listing the single trace primary invariants, ordered by increasing degree, by the time we reach $L_{N,d}$ at least one of the invariants listed must be a secondary invariant. 

Fix an all-single-trace homogeneous system of parameters and write
\begin{equation}
        \mathcal P=\mathbb C[p_1,\ldots,p_h].
\end{equation}
The secondary basis is isolated by setting all positive-degree primaries to zero. Setting all positive degree primaries to zero produces the graded quotient
\begin{equation}
        \mathcal Q_{\mathcal P}=\frac{\mathcal R_{N,d}}{(p_1,\ldots,p_h)\mathcal R_{N,d}}\cong\bigoplus_{\alpha=0}^{s-1}\mathbb C\,\overline{\eta}_\alpha.  \label{eq:secondary-quotient-basis},
\end{equation}
where we have used the free decomposition \eqref{eq:hironaka-decomposition}. In terms of the quotient, the degree of the first nontrivial secondary is given by
\begin{equation}
  \delta(\mathcal P)=\min\left\{m>0:(\mathcal Q_{\mathcal P})_m\neq0\right\}. \label{eq:first-secondary-degree}
\end{equation}
This degree depends, in general, on the chosen primary coordinates. Let
\begin{equation}
        \mathcal U_m={\rm span}\{p_a:\deg p_a=m\} \subseteq\mathcal T_m  \label{eq:used-primary-subspace}
\end{equation}
be the subspace of degree-$m$ single traces used as primary coordinates, and define
\begin{equation}
        m_*(\mathcal P)=\min\{m\geq1:\mathcal U_m\neq\mathcal T_m\}.\label{eq:first-missing-degree}
\end{equation}
From the above definition, $m_*$ is the smallest degree at which at least one single-trace invariant is not used as a primary invariant. We will now prove that this trace which is not used as a primary invariant is actually a secondary invariant.

\subsection{The first missing trace is the first secondary}\label{subsec:first-missing-trace}

The key input is the stable range of the trace algebra.  The first trace relation appears at degree $N+1$~\cite{Pr,Razmyslov,DrenskyFormanek}.  Up to this degree, distinct monomials in the single trace invariants are linearly independent.  In particular, a single trace cannot be cancelled by a linear combination of multi-trace operators of the same total degree.

\begin{theorem}[First missing single-trace direction] \label{thm:first-secondary}
Let $\mathcal P$ be an all-single-trace homogeneous system of parameters. If $m_*(\mathcal P)\leq N$, then $\delta(\mathcal P)=m_*(\mathcal P)$, and the first secondary space is canonically the unused part of the single-trace space,
\begin{equation}
(\mathcal Q_{\mathcal P})_{m_*} \cong \frac{\mathcal T_{m_*}}{\mathcal U_{m_*}}.\label{eq:first-secondary-quotient}
\end{equation}
Consequently,
\begin{equation}
\dim(\mathcal Q_{\mathcal P})_{m_*}=\dim\mathcal T_{m_*}-\dim\mathcal U_{m_*}.\label{eq:first-secondary-multiplicity}
\end{equation}
\end{theorem}

\begin{proof}
For every $m<m_*$ one has $\mathcal U_m=\mathcal T_m$.  Every invariant of total degree $m$ is a linear combination of products of single traces whose individual degrees are at most $m$.  Each positive-degree trace in such a product is thus a linear combination of primaries.  The entire positive-degree part of $\mathcal R_{N,d}$ below $m_*$ vanishes in the quotient \eqref{eq:secondary-quotient-basis}, and hence
\begin{equation}
   (\mathcal Q_{\mathcal P})_m=0,  \qquad 0<m<m_*.    \label{eq:no-secondary-below-mstar}
\end{equation}
At degree $m_*$, separate the invariant space into its single-trace and multi-trace parts.  Every genuine multi-trace monomial contains at least two positive-degree factors, so each factor has degree strictly less than $m_*$.  These factors lie in the span of lower-degree primaries, and the multi-trace monomial therefore vanishes in $\mathcal Q_{\mathcal P}$. Within the single-trace sector, the combinations that vanish directly are precisely the degree-$m_*$ primary combinations $\mathcal U_{m_*}$.  Since $m_*\leq N$, stable-range independence forbids an additional relation that would identify a remaining single trace with multi-trace terms.  The surviving space is therefore exactly
$\mathcal T_{m_*}/\mathcal U_{m_*}$, proving the theorem.
\end{proof}
The theorem gives a precise version of the statement that the first unused single trace becomes a secondary.  Because primaries can be linear combinations of words, the invariant statement concerns a missing direction in the vector space $\mathcal T_{m_*}$, not necessarily one distinguished word.

\subsection{Necklace counting and the logarithmic bound}\label{subsec:necklace-bound}

In the stable range, $\dim\mathcal T_m$ is the number of cyclic words of length $m$ in an alphabet of $d$ matrices\footnote{Throughout we assume that $d\ge 2$ so that $h=1+(d-1)N^2$. For $d=1$ we have $h=N$. For $d=1$ there are no secondary invariants and no overcrowding.}.  The counting of large-$N$ single-trace operators is equivalent to counting cyclic words, or necklaces, and is naturally performed using Burnside's lemma or Polya enumeration; see, for example~\cite{Sundborg:1999ue,Spradlin:2004pp,Bianchi:2003wx}. Burnside's lemma gives the necklace formula
\begin{equation}
a_d(m)\equiv\dim\mathcal T_m =\frac{1}{m}\sum_{r\mid m}\varphi(r)d^{m/r},  \qquad m\leq N,
\label{eq:necklace-formula}
\end{equation}
where $\varphi$ is Euler's totient function.  Let
\begin{equation}
        A_d(L)=\sum_{m=1}^{L}a_d(m) \label{eq:cumulative-single-traces}
\end{equation}
be the number of independent single-trace directions through length $L$, and define the \emph{overcrowding scale}
\begin{equation}
        L_{N,d}=\min\left\{L:A_d(L)>1+(d-1)N^2\right\}. \label{eq:overcrowding-scale}
\end{equation}
There are only $h=1+(d-1)N^2$ primary slots in total.  If every single-trace direction through degree $L$ is used as a primary, we would need at least $A_d(L)$ primaries.  The definition \eqref{eq:overcrowding-scale} therefore implies
\begin{equation}
        m_*(\mathcal P)\leq L_{N,d}  \label{eq:mstar-bound}
\end{equation}
for every all-single-trace homogeneous system of parameters. Combining this observation with Theorem~\ref{thm:first-secondary} gives the following universal bound.

\begin{corollary}[Logarithmic onset]\label{cor:logarithmic-onset}
If $L_{N,d}\leq N$, then every all-single-trace Hironaka decomposition obeys
\begin{equation}
        \delta(\mathcal P)\leq L_{N,d}.   \label{eq:delta-upper-bound}
\end{equation}
At fixed $d\geq2$ and large $N$,
\begin{equation}
        L_{N,d}=2\log_d N+\log_d\!\log_d N+O_d(1).  \label{eq:logarithmic-asymptotics}
\end{equation}
\end{corollary}

To obtain the large-$N$ estimate, note that
\begin{equation}
        a_d(m)=\frac{d^m}{m}+O_d\!\left(m d^{m/2}\right),
\end{equation}
and that the cumulative sum is dominated by its upper endpoint,
\begin{equation}
        A_d(L)=\frac{d^{L+1}}{(d-1)L}\left(1+O_d\!\left(\frac1L\right)\right).
        \label{eq:cumulative-asymptotic}
\end{equation}
Equating this expression to $h\simeq(d-1)N^2$ yields
\begin{equation}
        d^L\simeq\frac{(d-1)^2}{d}\,L N^2.
\end{equation}
Taking a logarithm once gives $L=2\log_dN+\log_dL+O_d(1)$; substituting the leading estimate $L\sim2\log_dN$ inside the second logarithm gives \eqref{eq:logarithmic-asymptotics}.  Since $\log N\ll N$, the stable-range condition is automatic for sufficiently large $N$ at fixed $d$.

For a concrete example, take two $10\times10$ matrices.  There are $h=N^2+1=101$ primary directions, while necklace counting gives
\begin{equation}
        A_2(8)=93,  \qquad A_2(9)=153.
\end{equation}
Thus every all-single-trace choice of primaries misses a single-trace direction by degree nine, and consequently contains a secondary of degree nine or less.  The first trace relation occurs only at degree eleven.  The secondary sector is therefore forced to appear before any trace relation is available to make two trace monomials linearly dependent.

\subsection{Dependence on the choice of primaries} \label{subsec:choice-dependence}

The bound \eqref{eq:delta-upper-bound} is not universal. A poor choice of primary coordinates can omit a low-degree direction and make the first secondary appear earlier. We will use a simple hypersurface example to illustrate this point.  Consider
\begin{equation}
        R=\frac{\mathbb C[a,b,c]}{(c^2-a^4-b^4)},\qquad \deg a=\deg b=1, \quad \deg c=2.
        \label{eq:hypersurface-example}
\end{equation}
With primaries $a,b$,
\begin{equation}
        R=\mathbb C[a,b]\oplus c\,\mathbb C[a,b],   \label{eq:hypersurface-decomp-one}
\end{equation}
so the first nontrivial secondary is $c$, of degree two.  With primaries $a,c$, the same ring has the decomposition
\begin{equation}
        R=\mathbb C[a,c] \oplus b\,\mathbb C[a,c] \oplus b^2\,\mathbb C[a,c] \oplus b^3\,\mathbb C[a,c],
        \label{eq:hypersurface-decomp-two}
\end{equation}
and the first secondary has degree one.  The ring is unchanged; only the polynomial coordinate algebra has changed.

It is therefore useful to distinguish the basis-dependent quantity $\delta(\mathcal P)$ from the latest onset attainable among all single-trace choices,
\begin{equation}
        \Delta_{N,d} = \sup_{\mathcal P\,\text{all single trace}}\delta(\mathcal P).\label{eq:optimal-secondary-onset}
\end{equation}
Our result implies
\begin{equation}
        \Delta_{N,d}\leq L_{N,d}
\end{equation}
in the stable range.  Whether one can always choose a degree-optimized single-trace system of parameters that saturates this bound is a sharper question, that can not be settled by the existence theorem we proved above.

\subsection{What changes at length \texorpdfstring{$O(\log N)$}{O(log N)}?}\label{subsec:physics-logN}

There are a number of important scales that are relevant to finite $N$ physics. The scale of the trace relations reflects the point at which individual traces of length $N+1$ become linearly dependent. This scale determines the giant graviton bound~\cite{McGreevy:2000cw}, fixing the maximal angular momentum of a giant graviton brane. Another important scale is the scale $N\log(N)$ at which large $N$ factorization fails~\cite{Garner:2014kna}. The overcrowding scale $L_{N,d}=O(\log N)$ is parametrically earlier than these finite-$N$ scales. What is the physical interpretation of the overcrowding scale? The overcrowding effect is collective and it is an algebraic precursor of finite-$N$ physics, distinct from the stringy exclusion mechanism at degree $O(N)$ or failing factorization at $O(N\log N)$. The word ``precursor'' is important. Since the secondary basis depends on the choice of primary coordinates, the appearance of a particular secondary at a particular degree is not, by itself, a basis-independent physical phenomenon. A dynamical interpretation of the overcrowding scale requires additional physical input. We take up this question in the next section.

\section{From Overcrowding to Fast Scrambling}\label{fastscrambling}

To explore the physical significance of the overcrowding scale, an important clue is that the overcrowding scale has precisely the logarithmic dependence on $N$ expected for the scrambling time of a fast scrambler~\cite{Sekino:2008he}. We will argue that this agreement is not accidental: under a natural assumption about operator growth, the counting argument that determines the overcrowding scale becomes the same argument that determines the scrambling time.

Overcrowding is kinematical. It follows from the structure of the multi-matrix invariant ring and is independent of the Hamiltonian. Scrambling, by contrast, is dynamical: a free or integrable Hamiltonian need not scramble information. To connect the two, we must restrict attention to a chaotic matrix model whose Hamiltonian is constructed from traces of words of bounded degree $r$.

Time evolution is generated by repeated commutators with the Hamiltonian. Since $H$ has bounded degree, a single commutator can increase the word length of an operator by only a bounded amount. Schematically,
\begin{equation}
    \deg[H,\cO]\leq \deg\cO+r-2.
\end{equation}
It follows that
\begin{equation}
    \deg\left(\operatorname{ad}_H^n\cO\right)\lesssim L_0+n(r-2),
\end{equation}
where $L_0$ is the initial word length. The Heisenberg expansion
\begin{equation}
    \cO(t)    =    \sum_{n=0}^{\infty}    \frac{(it)^n}{n!}{\rm ad}_H^n\cO
\end{equation}
therefore admits an operator front whose maximal word length grows at most linearly with time. We will assume that the dynamics saturates this bound, so that the characteristic operator length grows ballistically,
\begin{equation}
    L(t)\simeq v_{\rm op}t.    \label{eq:ballistic-growth}
\end{equation}
As discussed below, this is a reasonable assumption~\cite{Lieb:1972wy,Hastings:2005pr,Hastings:2010vzr,Roberts:2014isa} achieved by many systems. The counting results of the previous section then imply that the number of single-trace words available behind this front, denoted $K(t)$, grows exponentially,
\begin{equation}
    K(t)    \sim d^{L(t)}    \sim e^{\lambda_{\rm op}t},    \qquad    \lambda_{\rm op}=v_{\rm op}\log d.
    \label{eq:operator-spreading}
\end{equation}
For bosonic matrix models, making this statement precise requires control over operator norms and the energy sector under consideration. Equation~\eqref{eq:ballistic-growth} is therefore a dynamical assumption rather than a consequence of invariant theory alone.

The primary invariants provide the independent perturbative degrees of freedom. Their number is
\begin{equation}
    N_{\rm pert}=h=1+(d-1)N^2.
\end{equation}
Suppose that a new particle is created by acting with a simple single-trace operator. Initially, the information distinguishing the perturbed state from the original state is localized in only one, or a few, perturbative directions. Under chaotic evolution, the operator spreads through an increasing number of primary directions. We say that the information is scrambled once it can no longer be recovered from any parametrically small subset of the perturbative variables, so that access to an $O(1)$ fraction of the $N_{\rm pert}\sim N^2$ directions is required.

If the number of primary directions explored by the perturbation grows as
\begin{equation}
    n_{\rm spread}(t)\sim e^{\lambda_{\rm op}t},
\end{equation}
then the scrambling time is determined by
\begin{equation}
    n_{\rm spread}(t_*)\sim N_{\rm pert}\sim N^2.
\end{equation}
Consequently,
\begin{equation}
    t_*    \sim    \frac{1}{\lambda_{\rm op}}\log N_{\rm pert}    =    \frac{2}{\lambda_{\rm op}}\log N+O(1),
\end{equation}
which is the usual fast-scrambling scale.

The same estimate follows directly from overcrowding. The cumulative number of independent single-trace species through length $L$ behaves as
\begin{equation}
    A_d(L)\sim \frac{d^{L+1}}{(d-1)L},
\end{equation}
and the overcrowding scale is defined by
\begin{equation}
    A_d(L_{N,d})\sim N_{\rm pert}.
\end{equation}
Before this scale is reached, the growing operator explores only a parametrically small fraction of the perturbative directions. At $L=L_{N,d}$, the number of available single-trace directions becomes comparable to the complete set of primary invariants. The information carried by the added particle has then spread across an $O(1)$ fraction of the perturbative sector.

Combining this observation with the ballistic-growth assumption gives
\begin{equation}
    t_*\simeq\frac{L_{N,d}}{v_{\rm op}}=\frac{1}{v_{\rm op}}\left(2\log_d N+\log_d\log_d N+O_d(1) \right).
\end{equation}
Using $\lambda_{\rm op}=v_{\rm op}\log d$, this may be written as
\begin{equation}
    t_* = \frac{2\log N+\log\log N+O(1)}{\lambda_{\rm op}}.
\end{equation}
Thus the derivation of the scrambling time is almost identical to the derivation of the overcrowding scale. Both ask when an exponentially proliferating family of single-trace operators becomes comparable in number to the $O(N^2)$ independent perturbative degrees of freedom:
\begin{equation}
    \begin{aligned}
        \text{overcrowding:}\qquad&
        A_d(L_{N,d})\sim N_{\rm pert},
        \\
        \text{scrambling:}\qquad&
        n_{\rm spread}(t_*)\sim N_{\rm pert},
        \\
        \text{ballistic growth:}\qquad&
        L(t_*)\sim L_{N,d}.
    \end{aligned}
\end{equation}
Under these assumptions, overcrowding gives a microscopic algebraic description of the onset of scrambling: it is the point at which the growing operator has enough distinct single-trace components to explore the full perturbative sector.

The ballistic-growth assumption is motivated by studies of local chaotic systems. In particular, \cite{Roberts:2014isa} finds linear growth of the spatial support, and of the associated Pauli-string length, of precursor operators in a chaotic spin chain. Trace-word length in matrix quantum mechanics is not identical to spatial support, but the underlying mechanism is analogous: a bounded-degree Hamiltonian changes the word length by only a bounded amount in each nested commutator. It is therefore natural to expect a ballistic front in trace-word length in a dynamical matrix model, although this remains a dynamical assumption that should ultimately be tested in explicit examples.

Two final qualifications are important. First, the primary invariants are algebraically independent coordinates, but they are not automatically orthogonal subsystems or qubits. The statement that recovery requires access to an $O(1)$ fraction of the primaries therefore needs a dynamical notion of distinguishability. A natural choice is the inner product defined by matrix-model correlation functions. In the free large-$N$ theory, one may choose a basis of single-trace directions that is orthogonal to leading order, and we assume that interactions do not qualitatively invalidate this organization over the timescales of interest. Finally, the algebraic result by itself does not prove scrambling; it identifies the finite perturbative capacity that chaotic dynamics must fill. The additional dynamical claim is that operator evolution explores this space sufficiently democratically that saturation of the primary sector coincides with the loss of recoverability from every parametrically small primary subalgebra. Poetically, overcrowding is the algebraic shadow cast by fast scrambling.

\section{A toy model with \texorpdfstring{$\mathbb Z_2$}{Z2} gauge symmetry}\label{sec:z2}

The discussion of the secondary invariants given above has been algebraic. To understand what secondary invariants mean physically, we consider a deliberately simple toy model comprising of two primary invariants and a single secondary invariant. By studying this model we aim to isolate the kinds of effects we might look for in a genuine matrix model quantum mechanics.

Consider two real variables $x,y$ with a $\mathbb Z_2$ gauge identification
\begin{equation}
        (x,y)\sim(-x,-y).
\end{equation}
The gauge-invariant polynomials are generated by
\begin{equation}
        u=x^2,\qquad w=y^2,\qquad v=xy,   \label{eq:z2-invariants}
\end{equation}
subject to the relation, $v^2=uw$. Thus the invariant ring is
\begin{equation}
        R=\frac{\mathbb C[u,w,v]}{(v^2-uw)}.
\end{equation}
If we choose $u$ and $w$ as primary invariants, then
\begin{equation}
        R=\mathbb C[u,w]\oplus v\,\mathbb C[u,w].        \label{eq:z2-hironaka}
\end{equation}
The two secondary basis elements are $1$ and $v$. The primary data $u,w$ determines the magnitudes $|x|$ and $|y|$, but not the relative sign of $x$ and $y$.  For example,
\begin{equation}
        (x,y)=(1,1),\qquad (x,y)=(1,-1)
\end{equation}
have the same primary data $u=w=1$, but different values of the secondary,
\begin{equation}
        v=+1,\qquad v=-1.
\end{equation}
These two configurations are not related by a gauge transformation, which flips both signs simultaneously.  \emph{The secondary $v$ records a genuine piece of gauge-invariant information that is invisible to the chosen primaries.}  Geometrically, the map from invariant space to the primary-coordinate plane is a two-sheeted cover,
\begin{equation}
        v=\pm\sqrt{uw}. \label{eq:z2-two-sheets}
\end{equation}
The primary invariants describe the continuous base, while the non-trivial secondary invariant distinguishes the two sheets above a generic point of that base.

At this point we have described the configuration space of the toy model. Let us now use this configuration space for a 0-dimensional system, i.e. one whose path integral is an ordinary integral, and exhibit the same distinction at the semiclassical level.  Consider the zero-dimensional model with action given by
\begin{eqnarray}
        S_h(x,y)&=&N\left[\frac{\lambda}{4}(x^2-1)^2+\frac{\lambda}{4}(y^2-1)^2-hxy\right]\nonumber\\
                    &=&N\left[\frac{\lambda}{4}(u-1)^2+\frac{\lambda}{4}(w-1)^2-hv\right].
        \label{eq:z2-zero-dimensional-action}
\end{eqnarray}
We interpret $N^{-1}$ as $\hbar$, $\lambda$ as a coupling strength and $h$ as an external source that couples to the secondary. We can always write the action in terms of the invariants, as we did above, because the action itself must be gauge invariant. At $h=0$ there are four classical minima, $(x,y)=(\pm1,\pm1)$.  After dividing by the simultaneous sign flip, there are two gauge-inequivalent semiclassical saddles, denoted by $S_\pm$, distinguished by the value of the secondary invariant
\begin{equation}
        S_+:\quad v=+1,\qquad S_-:\quad v=-1.
\end{equation}
The primary data are the same on both saddles: $u=w=1$. Consequently, any observable built only from the primaries has the same leading semiclassical value in the two sectors,
\begin{equation}
\langle F(u,w)\rangle_\pm=\frac{1}{Z_\pm}\int_{S_{\pm}} dx\, dy F(x^2,y^2)e^{-S_0(x,y)}
\end{equation}
where
\begin{equation}
Z_{\pm}=\int_{S_{\pm}} dx\, dy e^{-S_0(x,y)}
\end{equation}
Thus, any observable depending only on the primary invariants is unable to distinguish the two sectors. On the other hand, the secondary invariant does distinguish sectors
\begin{equation}
   \langle v\rangle_+=1+O(N^{-1}),\qquad \langle v\rangle_-=-1+O(N^{-1}). \label{eq:z2-secondary-expectation}
\end{equation}
The source $h$ couples directly to the non-trivial secondary invariant.  For a fixed nonzero source the positions of the two saddles are shifted.  In the same-sign sector $S_+$, plugging $x=y=a_+$ into the equations minimizing the action, we find the position of the saddle is
\begin{equation}
        u_+=w_+=a_+^2=1+{h\over\lambda}, \qquad v_+=a_+^2=1+{h\over\lambda}, \label{eq:z2-shifted-same-sign-saddle}
\end{equation}
while in the opposite-sign sector $S_-$, writing $x=-y=a_-$ gives
\begin{equation}
        u_-=w_-=a_-^2=1-{h\over\lambda}, \qquad v_-=-a_-^2=-1+{h\over\lambda}. \label{eq:z2-shifted-opposite-sign-saddle}
\end{equation}
The two nonzero stationary points exist for $|h|<\lambda$, and both are local minima for sufficiently small $|h|$ (in particular, for $|h|<\lambda/2$).  Their classical action densities are
\begin{equation}
        S(h)\big|_{S_+}=-h-{h^2\over2\lambda}, \qquad S(h)\big|_{S_-}=h-{h^2\over2\lambda}.
        \label{eq:z2-shifted-saddle-actions}
\end{equation}
Thus, up to the fluctuation determinants, the two-saddle approximation gives
\begin{equation}
        Z(h)=\int\, dx\, dy e^{-S_h(x,y)}
        \simeq e^{Nh^2/(2\lambda)}\left(Z_+(h)e^{Nh}+Z_-(h)e^{-Nh}\right).\label{eq:z2-shifted-partition-function}
\end{equation}
where $Z_\pm(h)$ contain the Gaussian fluctuation determinant and higher-order perturbative corrections around the corresponding saddle. For a source that tends to zero at large $N$, the two fluctuation determinants agree at leading order,
\begin{equation}
        Z_+(h)=Z_-(h)+O(h),
\end{equation}
and the sector weights reduce to the familiar $\cosh(Nh)$ form.  Keeping the displacement of the saddles, one obtains
\begin{equation}
        \langle v\rangle_h = {h\over\lambda} + \tanh(Nh) + O(N^{-1}), \label{eq:z2-corrected-v-expectation}
\end{equation}
and
\begin{equation}
 \langle u\rangle_h=\langle w\rangle_h=1+{h\over\lambda}\tanh(Nh)+O(N^{-1}).\label{eq:z2-corrected-primary-expectations}
\end{equation}
Consequently, for fixed $h$ the primary observables are shifted by an amount of order $h/\lambda$. Now consider a situation in which the source $h_N$ vanishes as $N$ grows, but it satisfies
\begin{equation}
        h_N\longrightarrow0,  \qquad Nh_N\longrightarrow\pm\infty. \label{eq:z2-double-scaling-source}
\end{equation}
For example, we can take $h_N=cN^{-\alpha}$, $c\neq0$ and $0<\alpha<1$. With this source, at large $N$ we find
\begin{equation}
\langle v\rangle_{h_N}\longrightarrow {\rm sign}(c),\qquad \langle u\rangle_{h_N}, \langle w\rangle_{h_N}
\longrightarrow1.   \label{eq:z2-double-scaling-selection}
\end{equation}
This source selects a secondary sector, while its effect on the primary saddle positions disappears. 

This analysis gives the secondary invariant a concrete physical meaning: it distinguishes semiclassical sectors with identical primary data. At finite $N$ and $h=0$, the partition function sums equally over both sectors, so $\langle v\rangle=0$ even though $v$ is order one within each sector. A source $h$ biases their relative weights by $e^{2Nh}$, making $Nh$ the relevant parameter. In the double-scaling limit $h\to 0$ with $Nh\to\infty$, an asymptotically vanishing source selects a definite secondary sector without shifting the limiting primary configuration. Thus the secondary supplies the missing sector label while remaining invisible to leading-order primary observables.

The algebraic structure of the problem suggests that there may be a characteristic fluctuation signal. In addition to the usual fluctuations within each perturbative sheet, fluctuations may also arise from averaging over localized sectors. At $h=0$, the finite-$N$ integral sums symmetrically over the two sectors.  Therefore
\begin{equation}
        \langle v\rangle=0, \qquad \langle v^2\rangle\simeq1, \qquad {\rm Var}(v)\simeq1+O(N^{-1}).
        \label{eq:z2-large-variance}
\end{equation}
This order-one variance is not the ordinary width of a perturbative saddle, of order $N^{-1}$.  It arises because the path integral sums over two sharply localized sectors whose values of $v$ differ by two.  In the two-saddle approximation, the contribution due to an unresolved choice of sector is
\begin{equation}
        {\rm Var}_{\rm inter}(v)={\rm sech}^2(Nh).  \label{eq:z2-inter-sector-variance}
\end{equation}
The full variance also contains the ordinary fluctuation within each saddle, and hence
\begin{equation}
        {\rm Var}(v)={\rm sech}^2(Nh)+O(N^{-1}), \label{eq:z2-suppressed-variance}
\end{equation}
with source-dependent corrections to the $O(N^{-1})$ term when $h$ is held fixed.  Once a sector has been selected, the first term is exponentially suppressed,
\begin{equation}
        {\rm sech}^2(Nh)  \sim4e^{-2N|h|},
\end{equation}
and the ordinary intra-saddle variance of order $N^{-1}$ remains. An order-one variance that disappears upon sector selection is a signature of unresolved secondary-sector data.

The same distinction between sectors becomes nonperturbative in quantum mechanics.  Promote the toy model to quantum mechanical theory with action
\begin{equation}
S=N\int d\tau\left[T-V\right]\quad T=\frac12\dot x^2+\frac12\dot y^2\quad V=\frac{\lambda}{4}(x^2-1)^2+\frac{\lambda}{4}(y^2-1)^2.    \label{eq:z2-qm-action}
\end{equation}
The classical minima of the theory lie at the points $x=\pm1$ and $y=\pm1$. The values of the primaries at the minima is $u=1=v$. Taking into account the gauge symmetry, there are two inequivalent classical minima, again distinguished by their value of the secondary $v=+1$ and $v=-1$. To tunnel between them, one variable ($x$ or $y$) must change sign. The effects of tunneling are described by instanton solutions of the Euclidean equations of motion
\begin{equation}
S=N\int d\tau\left[T+V\right]
\end{equation}
Keeping, as an example, $x=1$, the exact instanton solution associated to tunnelling between $y=\pm 1$ is
\begin{equation}
        y_{\rm inst}(\tau)=\tanh\left[\sqrt\frac{\lambda}{2}\,(\tau-\tau_0)\right].
\end{equation}
The action of this instanton is
\begin{equation}
        S_{\rm inst}=N\int_{-1}^{1}dy\sqrt{2\,\frac{\lambda}{4}(y^2-1)^2}=\frac{2\sqrt{2\lambda}}{3}\,N,
        \label{eq:z2-instanton-action}
\end{equation}
so that the transition amplitude between the two secondary sectors behaves as
\begin{equation}
        \langle -|e^{-HT}|+\rangle \sim A\exp\left[-\frac{2\sqrt{2\lambda}}{3}\,N\right],
        \label{eq:z2-tunneling}
\end{equation}
up to the fluctuation prefactor $A$ which contains the Gaussian fluctuation determinant and higher-order perturbative corrections. In this toy model quantum mechanics the distinction measured by the secondary is invisible to every finite order in perturbation theory around a single saddle, mixing between sectors is exponentially suppressed at large $N$ and it can not be expanded as a series in $N^{-1}$.

The main achievement of the toy model has been to engineer situations in which the sectors associated with the secondaries really do correspond to distinct semiclassical sectors. This demonstrates that the algebraic construction furnished by the Hironaka decomposition can, under some circumstances, acquire dynamical significance. After this warm up we now turn from this toy example to the corresponding questions in simple matrix model quantum mechanics.

\section{Three matrix model}\label{sec:three-matrix}

Our goal in this section is to make the preceding discussion concrete in a genuine matrix model. We study a system of $d=3$ $N\times N$ matrices  with $N=2$. Although we are most interested in large values of $N$, the $N=2$ case is non-trivial and simple enough that explicit computations are possible. The invariant theory description of the relevant ring has been developed in~\cite{deMelloKoch:2025ngs,deMelloKochRodrigues:2026}. We will give some simple examples of actions that again realize the secondary sectors as distinct semiclassical sectors.

\subsection{Semiclassical sectors}

We have something very precise in mind when we talk about a semiclassical sector. A semiclassical sector is a region of configuration space organized around a classical saddle such that
\begin{itemize}
\item Perturbation theory can be developed locally around that saddle.
\item The saddle is separated from other relevant saddles by an action or energy barrier.
\item Transitions between the corresponding localized states are nonperturbative in the semiclassical parameter, typically of order $e^{-S_{\rm inst}/\hbar}$.
\end{itemize}
Clean examples of semiclassical sectors are provided, for example, by distinct degenerate or nearly degenerate classical minima, as given in Fig 1. below, More generally, they maybe true ground states, metastable ground states or more general long-lived classical configurations. We can also speak of semiclassical sectors associated with nontrivial classical solutions, but if such a solution has unstable directions, its more accurate to refer to it as a semiclassical \emph{saddle} rather than a \emph{sector}.
\bigskip

\bigskip

\bigskip

\begin{tikzpicture}
\begin{axis}[
    width=12cm,
    height=7.8cm,
    domain=-3:3,
    samples=500,
    xmin=-3,
    xmax=3,
    ymin=0,
    ymax=35,
    axis lines=left,
    xlabel={$x$},
    ylabel={$V(x)$},
    xlabel style={font=\large},
    ylabel style={font=\large},
    tick label style={font=\small},
    xtick={-3,-2,-1,0,1,2,3},
    ytick={0,10,20,30},
    clip=true,
]

\addplot[
    very thick,
    smooth,
]
{(x^2-1)^2*(0.7*x^2-4)^2};

\addplot[
    thin,
    opacity=0.35,
]
coordinates {(0,0) (0,85)};
\end{axis}
\end{tikzpicture}
\smallskip
\begin{quote}
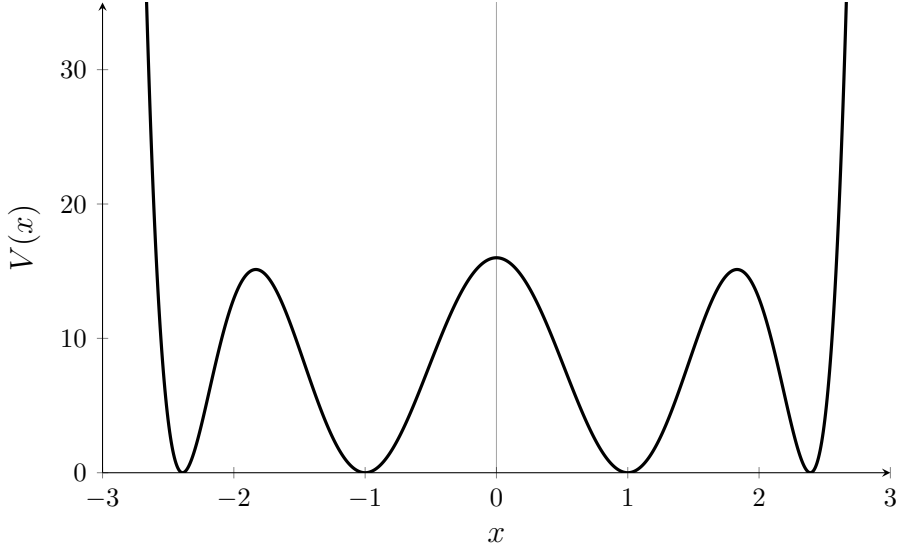
{\bf Figure 1. A four well potential:} A particle can be trapped in any given well. Perturbing it slightly causes the particle to oscillate about the bottom of the well and is described by using perturbation theory. Transitions between wells can be accomplished by quantum tunnelling. These effects are  typically of order $e^{-S_{\rm inst}/\hbar}$ with $S_{\rm inst}$ the instanton action.
\end{quote}
\bigskip

\subsection{Invariant Ring}

Following~\cite{deMelloKochRodrigues:2026} write the three Hermitian matrices as
\begin{equation}
        X^a={t_a\over 2}\,\mathbf 1+\vec x_a\cdot \vec \sigma,  \qquad t_a=\Tr X^a, \qquad
        \vec x_a\in \mathbb R^3, \qquad  a=1,2,3, \label{eq:pauli-decomposition-three}
\end{equation}
where $\vec\sigma$ are the Pauli matrices.  Under simultaneous conjugation by $U(2)$, the trace modes $t_a$ are invariant, while the traceless parts $\vec x_a$ transform by a common $SO(3)$ rotation.  Thus the invariant theory of three Hermitian $2\times 2$ matrices is the invariant theory of three vectors in $\mathbb R^3$, together with the three trace variables. The natural even invariants are the entries of the Gram matrix\footnote{A careful reader may complain that, with nonzero trace modes, $Q_{ab}=\frac12{\rm Tr}(X^aX^b)-\frac14{\rm Tr}(X^a){\rm Tr}(X^b)$, so that $Q_{ab}$ is not itself a single trace. This is not a substantive issue because $\mathbb C[t_a,Q_{ab}]=\mathbb C[t_a,{\rm Tr}(X^aX^b)]$, and the two hsops are related by a triangular polynomial transformation.}
\begin{equation}
        Q_{ab}=\vec x_a\cdot \vec x_b,  \qquad  a,b=1,2,3.  \label{eq:three-gram-invariants}
\end{equation}
These determine the lengths of the three vectors and their pairwise angles. They do not determine the orientation of the ordered triple. The missing orientation data is measured by the scalar triple product
\begin{equation}
        \chi=T_{123}=\vec x_1\cdot(\vec x_2\times \vec x_3). \label{eq:three-secondary-chi}
\end{equation}
In matrix language this is the imaginary part of a cubic trace,
\begin{equation}
        \operatorname{Im}\Tr(X^1X^2X^3)=2\chi . \label{eq:triple-product-matrix-form}
\end{equation}
The invariant $\chi$ changes sign under orientation reversal, while the $Q_{ab}$ do not. The Hilbert series is of the invariant ring is
\begin{equation}
        H_{2,3}(q)={1+q^3\over (1-q)^3(1-q^2)^6}. \label{eq:hilbert-three-matrix}
\end{equation}
An obvious and convenient choice of primary invariants $p_k$ $k=1,2,\cdots,9$ is
\begin{equation}
       p_k=\{t_1,t_2,t_3,Q_{11},Q_{22},Q_{33},Q_{12},Q_{13},Q_{23}\}. \label{eq:three-matrix-primaries}
\end{equation}
There is then one nontrivial secondary, namely $\chi$, and the Hironaka decomposition takes the simple form
\begin{equation}
        R_{2,3}={\cal P}\oplus \chi {\cal P},  \qquad {\cal P}=\mathbb C[t_a,Q_{ab}]. \label{eq:three-hironaka}
\end{equation}
The secondary is not a new free oscillator.  Its square is fixed by the primary data
\begin{equation}
\chi^2=\det G, \qquad G=\begin{pmatrix} Q_{11}&Q_{12}&Q_{13}\\ Q_{12}&Q_{22}&Q_{23}\\ Q_{13}&Q_{23}&Q_{33}
   \end{pmatrix}. \label{eq:chi-square-detG}
\end{equation}
For generic primary data with $\det G>0$, there are two secondary sectors
\begin{equation}
        \chi=+\sqrt{\det G}, \qquad \chi=-\sqrt{\det G}. \label{eq:three-two-sheets}
\end{equation}
The two sectors have identical primary data and differ only in the value of the orientation, determined by the secondary invariant. The two sectors meet on the discriminant locus
\begin{equation}
        \det G=0, \qquad \chi=0, \label{eq:three-discriminant}
\end{equation}
where the three vectors become linearly dependent.  Thus we again see that the secondary invariant records a discrete piece of gauge invariant data over a generic primary point.

\subsection{Matrix integrals and secondary expectation values}

We have seen that fixing the primaries does not fix the invariant configuration. A finite set of secondary choices remain. In this section we want to explore the consequences of this structure at the level of a zero-dimensional system. Towards this end consider a zero-dimensional matrix integral
\begin{equation}
        Z=\int \prod_{a=1}^3 dM_a\, e^{-S(M_1,M_2,M_3)} .    \label{eq:three-matrix-integral}
\end{equation}
After changing to invariant variables, the integral is a sum over the two secondary sectors, which are supported by the same primary data
\begin{equation}
        Z=\sum_{\epsilon=\pm 1}\int_{\mathcal D} d^9p_k\, J(p_k)\, \exp\left[-S\big(p_k,\epsilon\sqrt{\Delta(p_k)}\big)\right],\qquad\Delta(p_k)=\det G(p_k).  \label{eq:three-invariant-integral}
\end{equation}
Here $p_k$ denotes the nine primary invariants, $J(p_k)$ is the Jacobian associated to the change from matrix elements to the primary invariants, and $\mathcal D$ is the physical domain in which $G(p_k)$ is positive semidefinite. See~\cite{deMelloKochRodrigues:2026} for further details.

If the action is even under orientation reversal
\begin{equation}
        S(p,\chi)=S(p,-\chi), \label{eq:three-parity-even-action}
\end{equation}
then every observable built only from primaries is insensitive to the relative sign of $\chi$. In this case the two sheets contribute equally and supply a factor of two.  On the other hand, the secondary changes sign between sheets so that
\begin{equation}
        \langle\chi\rangle =0, \qquad  \langle \chi F(p)\rangle=0  \label{eq:odd-secondary-vanishes}
\end{equation}
for any function $F$ of the primaries.  The vanishing is a cancellation between the two orientation sheets.  The even moments do not vanish.  Since $\chi^2=\det G$,
\begin{equation}
        \langle\chi^2 F(p)\rangle=\langle\det G\, F(p)\rangle .  \label{eq:even-secondary-moments-three}
\end{equation}
Thus, a parity-symmetric state may have a vanishing one-point function for the secondary while retaining nonzero secondary fluctuations. Intuitively one might expect the secondary to exhibit an anomalously large variance, receiving contributions both from ordinary perturbative fluctuations within each sheet and from fluctuations between distinct sheets. A simple Gaussian example shows, however, that this expectation need not be realized. Specialize to the action
\begin{equation}
        S_0={\mu\over 2}\sum_{a=1}^3 \Tr M_a^2 ={\mu\over 2}\sum_{a=1}^3\left({t_a^2\over 2}+2|\vec x_a|^2\right).  \label{eq:three-gaussian-action}
\end{equation}
Since the vector components are independent Gaussians, we immediately find
\begin{equation}
\langle x_{a,i}x_{b,j}\rangle={1\over 2\mu}\delta_{ab}\delta_{ij}.\label{eq:gaussian-vector-variance}
\end{equation}
This implies the following expectation values for the Gram invariants,
\begin{equation}
        \langle Q_{aa}\rangle =\frac{3}{2\mu}, \qquad {\rm Var}(Q_{aa})=\frac{6}{4\mu^2},\qquad 
        \langle Q_{ab}\rangle=0, \qquad {\rm Var}(Q_{ab})=\frac{3}{4\mu^2}.  \label{eq:Qab-gaussian}
\end{equation}
The secondary obeys
\begin{equation}
        \langle \chi\rangle=0, \qquad \langle \chi^2\rangle=\frac{6}{8\mu^3}. \label{eq:chi-gaussian-second-moment}
\end{equation}
With the dimensionless normalization
\begin{equation}
        \widehat \chi = {\chi\over \big(\langle Q_{11}\rangle \langle Q_{22}\rangle \langle Q_{33}\rangle\big)^{1/2}},
        \label{eq:chi-normalized}
\end{equation}
one finds
\begin{equation}
    \langle \widehat\chi^2\rangle={2\over 9},\qquad {\rm Var}(\hat\chi)={2\over 9}. \label{eq:chi-normalized-variance}
\end{equation}
This is an ordinary normalized fluctuation.  The Gaussian model therefore shows that the algebraic existence of a secondary is not, by itself, a dynamical statement about large fluctuations.  Large secondary variance can arise only when the primary data are sharply localized while both sheets remain populated.  In that case the variance is dominated by the separation between the two sheets rather than by fluctuations within either sheet.

\subsection{Orientation instantons}

We now want to engineer a quantum mechanical system that realizes the secondary sectors as distinct semiclassical sectors. This is straightforward: any potential that fixes the primary invariants and is independent of the secondary, will realize the potential in both orientation sheets. For simplicity, specialize to traceless matrices $X^a=\vec x_a\cdot\vec\sigma$, with dynamics described by the action
\begin{equation}
   S=\int dt \left[{1\over 4}\sum_{a=1}^3 \Tr (\dot X^a\dot X^a)-V(X)\right],\label{eq:three-qm-action-real}
\end{equation}
with the following singe trace potential
\begin{eqnarray}
V&=&\frac{\lambda}{8}\sum_{a=1}^{3}\Tr\left(X^aX^a-v^2{\bf 1}\right)^2+\frac{\kappa}{16}\sum_{a<b}\Tr\left(\{X^a,X^b\}^2\right)\cr\cr
&=&{\lambda\over 4}\sum_{a=1}^3 \left(|\vec x_a|^2-v^2\right)^2+{\kappa\over 2}\sum_{a<b}
(\vec x_a\cdot \vec x_b)^2 . \label{eq:three-qm-potential}
\end{eqnarray}
The rewriting on the second line makes it manifest that this potential is only a function of the primary invariants. Since the two terms being summed are positive, the minima of the potential take the form
\begin{equation}
        |\vec x_a|=v,\qquad \vec x_a\cdot \vec x_b=0 \quad (a\neq b). \label{eq:three-qm-minima}
\end{equation}
Modulo common $SO(3)$ rotations, there are two minima, at $\chi=+v^3$ and $\chi=-v^3$. Thus, we have naturally obtained two classical minima. We would now like to explore the possibility that there is tunnelling between these two ground states. Towards that end, we need to find classical solutions of the Euclidean action
\begin{equation}
   S_E=\int d\tau \left[{1\over 4}\sum_{a=1}^3 \Tr (\dot X^a\dot X^a)+V(X)\right],\label{eq:three-qm-action}
\end{equation}
To construct an instanton solution, note that an explicit path connecting them is
\begin{equation}
        \vec x_1=v\vec e_1,  \qquad  \vec x_2=v\vec e_2, \qquad \vec x_3=q(\tau)\vec e_3 .
        \label{eq:three-instanton-ansatz}
\end{equation}
Plugging this ansatz into the equations of motion obtained from the Euclidean action \eqref{eq:three-qm-action}, we obtain the following equation of motion
\begin{equation}
{\partial^2\over\partial\tau^2}q=\lambda(q^2-v^2)q\,.
\end{equation}
The exact instanton solution is
\begin{equation}
        q_{\rm inst}(\tau) = v\tanh\!\left[ v\sqrt{\lambda\over 2}\,(\tau-\tau_0) \right]. \label{eq:three-instanton}
\end{equation}
Along this trajectory
\begin{equation}
 \chi_{\rm inst}(\tau)=\vec{x}_1\cdot(\vec{x}_2\times\vec{x}_3)=v^2 q_{\rm inst}(\tau), \label{eq:chi-instanton-profile}
\end{equation}
so the instanton interpolates from $-v^3$ to $+v^3$.  At its center $\chi=0$, and the trajectory passes through the rank-deficient locus where the three vectors become linearly dependent.  The instanton action is
\begin{equation}
        S_{\rm inst} = \int_{-v}^{v}dq\, \sqrt{2\,{\lambda\over 4}(q^2-v^2)^2} = {2\sqrt{2}\over 3}\,v^3\sqrt{\lambda}.
        \label{eq:three-instanton-action}
\end{equation}
Thus the algebraic branch locus has a direct dynamical role: it is the region through which the nonperturbative trajectory passes. Let $|+\rangle$ and $|-\rangle$ denote states localized near the two orientation minima.  The dominant contributions to transitions between these two sectors
\begin{equation}
\langle\pm |e^{-HT}|\mp\rangle\sim Ae^{-S_{\rm inst}}
\end{equation}
where $A$ is obtained from Gaussian fluctuations and higher loop corrections. In the semiclassical limit where the instanton action $S_{\rm inst}$ becomes large, mixing between these sectors is highly suppressed and we obtain genuine semiclassical sectors. The secondary invariant is the natural transition operator between the two sectors.

Given how naturally the secondary sectors emerge as genuine semiclassical sectors of an explicit dynamical system, we regard this as evidence that the algebraic sector structure may often be realized dynamically, as envisaged in~\cite{deMelloKoch:2025ngs}.

\section{Four matrix model}\label{sec:four-matrix}

The three-matrix model studied in the previous section provides a natural dynamical realization of secondary sectors as distinct semiclassical branches whose mixing is suppressed in the classical limit. That example, however, contains only a single secondary invariant, which records the orientation. We now turn to the $N=2$, $d=4$ model, whose structure is qualitatively richer. In this case, the secondary invariants arise both as orientation labels and as labels distinguishing the different solutions of an algebraic equation over the primary data. This more closely reflects the generic structure of the finite-$N$ invariant ring. In this case we will again see that it is not difficult to find dynamical systems that realize the secondary sectors as genuine semiclassical sectors. The invariant ring for this system has been studied in~\cite{deMelloKoch:2025ngs,deMelloKochRodrigues:2026}. 

Following~\cite{deMelloKochRodrigues:2026} we again expand
\begin{equation}
   X_a={t_a\over 2}\,\mathbf 1+\vec x_a\cdot\vec\sigma,\qquad a=1,\ldots,4. \label{eq:pauli-decomposition-four}
\end{equation}
There are four linear primary invariants $t_a$. The traceless parts of the matrices give four vectors in $\mathbb R^3$. Parallel to the analysis of the last section, their even invariants are the ten Gram-matrix entries
\begin{equation}
        Q_{ab}=\vec x_a\cdot \vec x_b, \qquad 1\leq a\leq b\leq 4, \label{eq:four-gram-invariants}
\end{equation}
and the single secondary of the last section is generalized to four orientation-sensitive cubic invariants, given by the four scalar triple products
\begin{equation}
        T_{123},\quad T_{124},\quad T_{134},\quad T_{234},\qquad
        T_{abc}=\vec x_a\cdot(\vec x_b\times \vec x_c). \label{eq:four-triple-products}
\end{equation}
It is not difficult to check that these cubic invariants share one overall orientation. The Hilbert series of the invariant ring is~\cite{deMelloKoch:2025ngs}
\begin{equation}
        H_{2,4}(q)={1+q^2+4q^3+q^4+q^6\over(1-q)^4(1-q^2)^9}. \label{eq:hilbert-four-matrix}
\end{equation}
It exhibits thirteen primary invariants: the four traces and nine quadratic invariants. Since there is a total of ten distinct quadratic invariants in the Gram entries, the overcrowding scale is saturated and a single quadratic invariant is realized in the secondary module. We need to make a concrete choice for the primary invariants and then determine the resulting secondary invariants.

\subsection{A quadratic primary choice}\label{subsec:four-primary-choice}

A valid choice for the set of primary invariants is given by
\begin{equation}
p_1=t_1,\quad p_2=t_2,\quad p_3=t_3,\quad p_4=t_4
\end{equation}
as well as 
\begin{eqnarray}
 &p_5=Q_{11},\quad p_6=Q_{22},\quad p_7=Q_{33},\quad p_8=Q_{44},\nonumber\\
 &p_9=Q_{13},\quad p_{10}=Q_{14},\quad p_{11}=Q_{23},\quad p_{12}=Q_{24},\nonumber \\
 &p_{13}=Q_{12}-Q_{34}. \label{eq:four-quadratic-primaries}
\end{eqnarray}
The remaining independent quadratic invariant $ \Sigma=\frac12(Q_{12}+Q_{34})$ must be included\footnote{Recall that all words with a length $\le N=2$ are linearly independent and must be included amongst the generating invariants.} in the set of generating invariants as a secondary invariant. To see that this is a valid choice, notice that if all nine $p_i$ vanish, we have
\begin{equation}
 G=\Sigma H, \qquad
 H=\begin{pmatrix}
 0&1&0&0\\
 1&0&0&0\\
 0&0&0&1\\
 0&0&1&0
 \end{pmatrix}.
 \label{eq:four-zero-gram}
\end{equation}
where $G$ is the Gram matrix $(G)_{ab}=Q_{ab}$. Because $\det H=1$ and a Gram matrix of four vectors in $\mathbb C^3$ has $\det G=0$, we find
\begin{equation}
 0=\det G=\Sigma^4, \label{eq:four-zero-sigma}
\end{equation}
so $\Sigma=0$.  Thus the simultaneous vanishing of the nine proposed primaries leaves only the origin in the Gram-invariant space, which confirms that these nine quadratic invariants form a valid homogeneous system of parameters.

The fact that a Gram matrix of four vectors in $\mathbb C^3$ has $\det G=0$ implies the following algebraic equation
\begin{equation}
\det G=\Sigma^4+f_1(p_k)\Sigma^3+f_2(p_k)\Sigma^2+f_3(p_k)\Sigma+f_4(p_k)=0\label{secondaryalgebraic}
\end{equation}
where the coefficients $f_a(p_k)$ are polynomials of degree $2a$ in the primary invariants. Thus, after choosing a definite set of primary data there are four possible values for the secondary $\Sigma$. In terms of $\Sigma$ the secondary invariants can be chosen as follows
\begin{eqnarray}
 &\eta_0=1,\quad \eta_1=\Sigma,\quad \eta_2=\Sigma^2,\quad \eta_3=\Sigma^3,\nonumber\\
 &\eta_4=T_{123},\qquad \eta_5=T_{124},\qquad \eta_6=T_{134},\quad \eta_7=T_{234},\label{secondariesValid}
 \end{eqnarray}
 For a given choice of primary data, we still need to choose one of four possible Gram roots\footnote{It is possible that not all roots are valid choices -- see the next section.} of \eqref{secondaryalgebraic} to determine $\eta_1$, $\eta_2$ and $\eta_3$. Once a set of primary data has been chosen, the remaining ambiguity in the cubic secondaries is the common reversal
\begin{equation}
 T_I\longmapsto-T_I \qquad\text{for all }I. \label{eq:four-orientation-reversal}
\end{equation}
so there are two choices for the remaining secondaries $\eta_4,\eta_5,\eta_6$ and $\eta_7$. It is immediately evident that the secondary multiplication law \eqref{eq:secondary-multiplication} is obeyed. Concretely, from \eqref{secondaryalgebraic} we immediately have
 \begin{eqnarray}
&\eta_3\eta_3=(2f_1f_2-f_1^3-f_3)\eta_3+(f_1f_3+f_2^2-f_1^2f_2-f_4)\eta_2\nonumber\\
&\qquad+(f_1f_4+f_2f_3-f_1^2f_3)\eta_1+(f_2f_4-f_1^2f_4)\eta_0,\nonumber\\
 &\eta_2\eta_3=(f_1^2-f_2)\eta_3+(f_1f_2-f_3)\eta_2+(f_1f_3-f_4)\eta_1+f_1f_4\eta_0,\nonumber\\
 & \eta_1\eta_1=\eta_2,\qquad\eta_1\eta_2=\eta_3,\qquad\eta_2^2=\eta_1\eta_3=-f_1\eta_3-f_2\eta_2-f_3\eta_1-f_4\eta_0,
 \end{eqnarray}
 where the $f_a$ are read from \eqref{secondaryalgebraic}. Products of cubic secondaries are fixed by Gram minors.  For example, the product law for $\eta_4^2$ follows from
\begin{eqnarray}
\eta_4^2&=& T_{123}^2=
 \det\begin{pmatrix}
 Q_{11}&Q_{12}&Q_{13}\\
 Q_{12}&Q_{22}&Q_{23}\\
 Q_{13}&Q_{23}&Q_{33}
 \end{pmatrix}=
 \det\begin{pmatrix}
 p_5&\frac12 p_{13}+\eta_1&p_9\\
 \frac12 p_{13}+\eta_1&p_6&p_{11}\\
 p_9&p_{11}&p_7
 \end{pmatrix},\nonumber\\
 &=& (p_{11}p_{13}p_9-p_{11}^2 p_5-\frac14 p_{13}^2 p_7+p_5 p_6 p_7-p_6 p_9^2)\eta_0+(2p_{11}p_9-p_{13}p_7) \eta_1-p_7 \eta_2 .\nonumber
\end{eqnarray}
More generally, any product between $\eta_4,\eta_5,\eta_6$ and $\eta_7$ can be reduced in eactly the same way using
\begin{equation}
 T_{ijk}T_{\ell mn} =\det(Q_{ab})_{a\in\{i,j,k\},\,b\in\{\ell,m,n\}}.\label{eq:four-triple-products}
\end{equation}
To complete the discussion, we need to consider the product of a Gram root and a cubic invariant, for example, $\eta_1\eta_4$.
 The product in question is therefore
\begin{equation}
 \eta_1\eta_4=\Sigma T_{123}.
\end{equation}
We now reduce this expression to the same secondary basis with coefficients in the primary ring. The key observation is that because the four vectors $\vec x_1,\ldots,\vec x_4$ lie in three dimensions, they satisfy the linear-dependence identity
\begin{equation}
 T_{234}\vec x_1 -T_{134}\vec x_2 +T_{124}\vec x_3 -T_{123}\vec x_4 =0. \label{eq:vector-identity}
\end{equation}
Taking the scalar product of equation~\eqref{eq:vector-identity} with $\vec x_3$, and performing some trivial manipulations gives
\begin{equation}
 \left(\Sigma-\frac{p_{13}}{2}\right)T_{123} =p_9T_{234}-p_{11}T_{134}+p_7T_{124}.
\end{equation}
This immediately implies
\begin{equation}
 \eta_1\eta_4=\frac{p_{13}}{2}\eta_4+p_7\eta_5-p_{11}\eta_6+p_9\eta_7. \label{eq:final-eta-product}
\end{equation}
To argue that the choice \eqref{secondariesValid} is valid, note that the proposed nine quadratic invariants, together with the four linear traces, form an hsop. Since the ring is Cohen–Macaulay, it is free over this polynomial algebra. Multiplying the Hilbert series by the hsop denominator gives $1+q^2+4q^3+q^4+q^6$. The proposed secondary invariants have precisely the degrees: $0,2,3,3,3,3,4,6$. It remains only to note that their residue classes are nonzero and linearly independent in $R/(p_1,…,p_{13})$.

For generic complex primary data, $\det G=0$ is quartic in $\Sigma$.  Each Gram root is doubled by the common orientation sign, giving eight complex points. Thus the invariant ring is described as an eight sheeted cover of the primary base with different sheets resolved by the eight secondaries. On the Hermitian slice however, $G$ must be positive semidefinite. Generically, at most two Gram roots are physical, giving at most four real points after orientation doubling.  The degree-eight statement therefore refers to the complexified fiber, not to eight Hermitian sheets.

The fact that not all eight sheets intersect the original Hermitian matrix integration contour does not imply that they are physically irrelevant. Instanton configurations themselves generally arise only after analytic continuation of the real time path integral, yet they can have genuine physical effects. By the same logic, sheets that do not lie on the original integration contour may still influence the theory through its complexified saddle-point structure.

\subsection{An explicit physical four-point fiber}\label{subsec:four-explicit-fiber}

In this section, we fix the primary data and determine the residual discrete information encoded by the secondary invariants, which explicitly illustrates the arguments developed in the previous section. Fix $Q_{11}=1$, $Q_{22}=2$, $Q_{33}=1$, $Q_{44}=2$, with $Q_{13}=Q_{14}=Q_{23}=0$, $Q_{24}=1$ and $Q_{12}-Q_{34}=0$. Then $Q_{12}=Q_{34}=\Sigma$ and
\begin{equation}
 G(\Sigma)=\begin{pmatrix}
 1&\Sigma&0&0\\
 \Sigma&2&0&1\\
 0&0&1&\Sigma\\
 0&1&\Sigma&2
 \end{pmatrix}.
 \label{eq:four-explicit-gram}
\end{equation}
Its determinant is
\begin{equation}
 \det G(\Sigma)=(\Sigma^2-1)(\Sigma^2-3). \label{eq:four-det-factor}
\end{equation}
The four complex roots are $\Sigma=\pm1,\pm\sqrt3$. Any such root is physical if the resulting Gram matrix is positive semidefinite. If a matrix is positive semidefinite, all of its principal minors are positive semidefinite. Consequently, the roots $\pm\sqrt3$ are not physical because
\begin{equation}
 \det\begin{pmatrix}Q_{11}&Q_{12}\\Q_{12}&Q_{22}\end{pmatrix} =2-\Sigma^2=-1.
 \label{eq:four-unphysical-minor}
\end{equation}
By contrast, $\Sigma=\pm1$ give rank-three positive-semidefinite Gram matrices, so that these roots are physical. We can choose an explicit set of vectors $\vec{x}_a$ in \eqref{eq:pauli-decomposition-four} that generate the primary data. Of course any gauge transformation of these representatives is an equally valid choice. Our choice
\begin{equation}
 \vec x_1^{(s,\epsilon)}=s\vec e_1, \qquad \vec x_2^{(s,\epsilon)}=\vec e_1+\epsilon\vec e_3, \qquad
 \vec x_3^{(s,\epsilon)}=s\vec e_2, \qquad \vec x_4^{(s,\epsilon)}=\vec e_2+\epsilon\vec e_3.
 \label{eq:four-representatives}
\end{equation}
is labelled by the pair $s,\epsilon$ which both take value $\pm1$. With these explicit vectors it is simple to verify that
\begin{equation}
 Q_{12}=Q_{34}=s, \qquad \Sigma=s. \label{eq:four-s-label}
\end{equation}
Changing $\epsilon$ leaves the complete Gram matrix unchanged and reverses all four cubic invariants.  The physical fiber is therefore
\begin{equation}
 (s,\epsilon)\in\{+1,-1\}\times\{+1,-1\}. \label{eq:four-physical-fiber}
\end{equation}
The four configurations $\vec x_1^{(s,\epsilon)},\vec x_2^{(s,\epsilon)},\vec x_3^{(s,\epsilon)},\vec x_4^{(s,\epsilon)}$ are gauge inequivalent.

\subsection{Gram-sheet and orientation instantons}\label{subsec:four-instantons}

We again would like to engineer a quantum mechanical system that realizes the secondary sectors as distinct semiclassical sectors. For the four matrix model we would like the dynamical system to exhibit two types of instantons: orientation-changing and Gram-sheet changing instantons. It again proves to be straightforward to find a dynamical system, that realizes both types of instantons. For simplicity we again consider  traceless matrices whose dynamics is governed by the action
\begin{equation}
 S_E=\int\dd\tau\left[ \frac14\sum_{a=1}^{4}\Tr\dot X_a^2+V\right]. \label{eq:four-euclidean-action}
\end{equation}
We choose the positive primary-only single trace potential
\begin{eqnarray}
V&=&{\rm Tr}\Bigg[
\frac{\lambda}{8}\left(X_1^2-\mathbf 1\right)^2+\frac{\lambda}{8}\left(X_3^2-\mathbf 1\right)^2+\frac{\lambda}{16}\left(\{X_2,X_4\}-\frac{1}{2}\left(X_2^2+X_4^2\right)\right)^2\nonumber\\[2mm]
&&\,\,\,\,+\frac{\lambda}{32}\left(X_2^2+X_4^2-4\mathbf 1\right)^2+\frac{\eta}{4}\left(X_2^2-X_4^2\right)^2
+\frac{\kappa}{16}\Big(\{X_1,X_2\}-\{X_3,X_4\}\Big)^2\nonumber\\[2mm]
&&\,\,\,\,+\frac{\rho}{16}\Big(\{X_1,X_3\}^2+\{X_1,X_4\}^2+\{X_2,X_3\}^2\Big)\Bigg]\nonumber\\[2mm]
&=&\frac\lambda4\left[(|\vec{x}_1|^2-1)^2+(|\vec{x}_3|^2-1)^2\right] +\frac\lambda2\left(B-\frac A4\right)^2
 +\frac\lambda{16}A^2 +\frac\eta2(|\vec{x}_2|^2-|\vec{x}_4|^2)^2\nonumber\\[2mm]
 &&\,\,\,\, +\frac\kappa2(\vec{x}_1\cdot\vec{x}_2-\vec{x}_3\cdot\vec{x}_4)^2
 +\frac\rho2((\vec{x}_1\cdot\vec{x}_3)^2+(\vec{x}_1\cdot\vec{x}_4)^2+(\vec{x}_2\cdot\vec{x}_3)^2), \label{eq:four-primary-potential}
\end{eqnarray}
with positive couplings and where we have defined
\begin{equation}
 A=|\vec{x}_2|^2+|\vec{x}_4|^2-4, \qquad B=\vec{x}_2\cdot\vec{x}_4-1, \label{eq:four-A-B}
\end{equation}
Notice that we the potential depends only on $\vec{x}_1\cdot\vec{x}_2-\vec{x}_3\cdot\vec{x}_4$ and not on $\vec{x}_1\cdot\vec{x}_2$ and $\vec{x}_3\cdot\vec{x}_4$ separately, ensuring it does not depend on any secondary invariants. Although this potential looks somewhat complicated, it has been engineered to be positive and to admit
\begin{equation}
|\vec{x}_1|=|\vec{x}_3|=\vec{x}_2\cdot\vec{x}_4=1,\qquad |\vec{x}_2|=|\vec{x}_4|=\sqrt{2},\qquad
\vec{x}_1\cdot\vec{x}_3=\vec{x}_1\cdot\vec{x}_4=\vec{x}_2\cdot\vec{x}_3=0
\end{equation}
as its physical minima. These are precisely the four points~\eqref{eq:four-physical-fiber} we studied in the last subsection. Consequently we know that this physical minimum appears in four sheets of the physical fiber. Thus, there are four classical minima parametrized exactly as in \eqref{eq:four-representatives}. We will now search for instanton transitions between them. Instanton solutions solve the equations of motion of the Euclidean action 
\begin{equation}
 S_E=\int\dd\tau\left[ \frac14\sum_{a=1}^{4}\Tr\dot X_a^2+V\right]. \label{eq:four-euclidean-action}
\end{equation}
For the Gram-sheet instanton, fix $\epsilon$ and consider the following path that takes us from one Gram root to another
\begin{equation}
 \vec x_1=q(\tau)\vec e_1, \quad \vec x_2=\vec e_1+\epsilon\vec e_3, \quad \vec x_3=q(\tau)\vec e_2, \quad
 \vec x_4=\vec e_2+\epsilon\vec e_3. \label{eq:four-gram-ansatz}
\end{equation}
With this choice for the $\vec{x}_a$, we have $\Sigma=q$. Plugging this ansatz into the equations of motion, we find the exact solution is
\begin{equation}
 q_{\rm Gram}(\tau)= \tanh\left[\sqrt{\frac\lambda2}(\tau-\tau_0)\right], \label{eq:four-gram-instanton}
\end{equation}
The action of the instanton is given by
\begin{equation}
 S_{\rm Gram}=\frac{4\sqrt{2\lambda}}3. \label{eq:four-gram-action}
\end{equation}
It changes $s$ at fixed $\epsilon$ and passes through a rank-two Gram matrix.

For the orientation instanton, a path leading from one orientation to the opposite is given by
\begin{equation}
 \vec x_a(\tau)=z(\tau)\vec x_a^{(s,+)}. \label{eq:four-orientation-ansatz}
\end{equation}
The endpoints $z=\pm1$ have the same Gram matrix, while
\begin{equation}
 T_{abc}(\tau)=z(\tau)^3T_{abc}^{(s,+)} \label{eq:four-orientation-triples}
\end{equation}
reverses sign. Again plugging this ansatz into the equations of motion we find that the exact solution is
\begin{equation}
 z_{\rm orient}(\tau)= \tanh\left[\sqrt{\frac\lambda2}(\tau-\tau_0)\right]\label{eq:four-orientation-instanton}
\end{equation}
and the corresponding action is
\begin{equation}
 S_{\rm orient}=4\sqrt{2\lambda}. \label{eq:four-orientation-action}
\end{equation}
The model therefore realizes two different finite labels: $\Sigma$ distinguishes physical Gram roots, while the cubic secondaries distinguish the orientation doubling. Once again this gives a dynamical system that realizes the interpretation suggested by the algebraic structure of the invariant ring.

\subsection{Secondary Invariants as semiclassical sectors}

Our goal was to give evidence supporting the interpretation of the secondary invariants as labels for semiclassical sectors. Have we done that? Fixing primaries leaves a finite set of distinct gauge-invariant configurations with identical primary data. For a concrete illustration, see Fig. 2 below.

\bigskip

\begin{center}
\begin{tikzpicture}[>=Latex,scale=1.5]
  \draw[->] (-3.4,-1.5)--(3.4,-1.5) node[right] {primary space};
  \foreach \y/\c in {1.1/blue,0.35/blue,-0.4/blue,-1.15/blue}{
    \draw[thick,\c] (-2.8,\y) .. controls (-1.2,\y+0.25) and (1.2,\y-0.25) .. (2.8,\y);}
  \draw[dashed] (0,-1.5)--(0,1.35);
  \node[right] at (0.15,0.75) {finite fiber};
  \node[below] at (0,-1.5) {fixed $p$};
\end{tikzpicture}
\end{center}
\begin{quote}
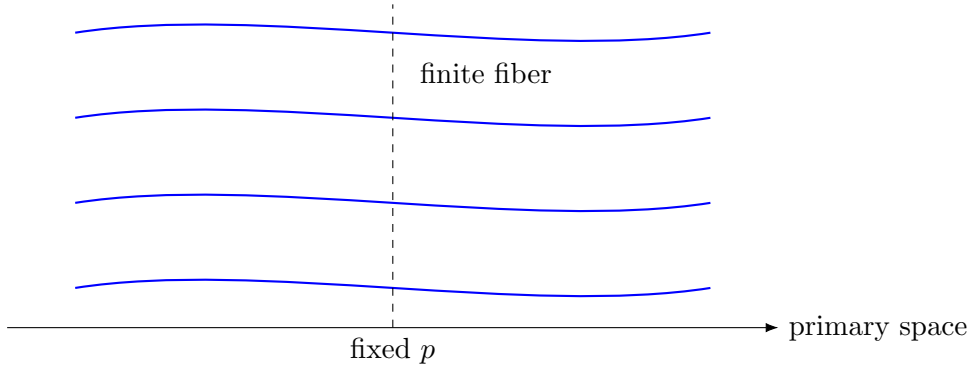
{\bf Figure 2. Secondary Sheets over the Primary Space} The primary data generally do not fix the secondary invariants, so the invariant configuration space forms a finite collection of sheets over the primary-coordinate space, distinguished by the secondary values. If the potential depends only on the primaries and has a minimum at a given primary point, that minimum is replicated on each dynamically accessible sheet. See the dashed line. The Hironaka decomposition therefore provides a natural algebraic mechanism for multiple degenerate ground states.
\end{quote}
\medskip

The figure above is an illustration of the matrix model example we have described in this section. We found the ground state appeared in all four sheets of the invariant space.  To show that these fiber points behave as distinct semiclassical sectors, we have completed the following natural dynamical tests

\begin{itemize}
\item \emph{Classical realization:} All four points are a minimum of the potential.

\item \emph{Stability:} We have checked, with the Hessian, that each is a local minimum. There are no zero energy valleys which would allow a particle to escape. A continuous path between two distinct ground states must leave the minimum locus and cross a region of larger potential or Euclidean action.

\item \emph{Degeneracy in primary observables:} We have verified that the configurations share identical primary data but differ in one or more secondary invariants. Their coexistence at fixed primary data is therefore a direct manifestation of the secondary structure.

\item \emph{Instanton interpolation:} We have explicitly constructed finite-action Euclidean solutions interpolating between configurations. Their actions determine exponentially small mixing between the localized semiclassical states.
\end{itemize}
Together, these tests give a concrete meaning to the statement that secondary invariants are associated with semiclassical sectors.

It is natural to ask how far this physical interpretation extends beyond the $N=2$ examples studied here, since the more interesting regime is large but finite $N$. Two ingredients are essential: the potential should fix the primary data, and it should be independent of the secondary invariants. As $N$ increases, the overcrowding scale moves to larger degree, so a bounded-degree potential will typically be constructed entirely from low-degree single traces and hence from primary invariants alone. If such a potential has isolated classical minima in the primary space, each dynamically accessible point in the finite fiber above a minimum provides a corresponding semiclassical ground state. Under these mild assumptions, the picture developed above should persist at larger finite $N$, with secondary invariants labeling distinct semiclassical sectors.

\section{Discussion}\label{sec:discussion}

The central result of this paper is that the finite-$N$ invariant algebra develops a non-free organization parametrically before the first trace identities appear. We proved that the primary invariants may always be chosen entirely from the single-trace sector. In the stable range, the first single-trace direction omitted from such a primary algebra survives in the quotient by the primaries and therefore becomes the first nontrivial secondary invariant. Combining this result with necklace counting gives
\begin{equation}
 \delta(\mathcal P) \leq L_{N,d}= 2\log_d N+\log_d\log_d N+O_d(1)
\end{equation}
for every all-single-trace homogeneous system of parameters $\mathcal P$. The secondary sector is therefore forced to appear at degree no greater than $O(\log N)$, even though all trace monomials at this degree remain linearly independent.

The overcrowding scale is a capacity bound rather than a determination of the exact onset. If $\Delta_{N,d}$ denotes the latest possible first secondary among all all-single-trace choices of primaries, our result establishes
\begin{equation}
 \Delta_{N,d}\leq L_{N,d}.
\end{equation}
Large-$N$ freeness suggests that the unavoidable onset cannot remain at fixed degree as $N$ grows, so that one expects
\begin{equation}
 1\ll \Delta_{N,d}\leq L_{N,d}=O(\log N),
\end{equation}
although a quantitative lower bound remains to be proved. The meaning of this logarithmic scale is global. It does not imply that an individual trace of length $O(\log N)$ has become dependent, nor that the planar expansion of a fixed low-degree observable has failed. Rather, there are too many independent single-trace directions for all of them to serve simultaneously as freely generated coordinates of the complete finite-$N$ invariant algebra.

This delayed onset is a distinctive feature of matrix models rather than a generic property of Hironaka decompositions. In vector models, and in bosonic systems with an $S_N$ permutation symmetry, nontrivial secondary data do appear at much lower degree. For example, for $N$ bosons in two dimensions, separate power sums in the two coordinate directions determine the unordered sets of $x$- and $y$-coordinates but not their pairing. Mixed invariants therefore supply secondary data already at low degree, and the generic fiber contains $N!$ sheets corresponding to the possible pairings. What is special in the matrix problem is the exponential proliferation of cyclic words with their length, together with only $O(N^2)$ algebraically independent continuous coordinates.
It is this specifically matrix-theoretic competition that produces the parametrically delayed scale $L_{N,d}=O(\log N)$.

The coincidence of this scale with the fast-scrambling scale is suggestive. If operator word length grows ballistically and chaotic evolution explores the available single-trace directions sufficiently democratically, the operator reaches the overcrowding scale precisely when the number of accessible single-trace directions becomes comparable to the $O(N^2)$ perturbative coordinates. Under these dynamical assumptions, the counting argument for overcrowding becomes the same counting argument that gives a fast scrambling time of order $\log N$. Overcrowding does not by itself prove scrambling: it is a kinematical statement about the finite capacity of the primary algebra. The proposed connection requires chaotic dynamics, ballistic operator growth and an appropriate dynamical notion of distinguishability among the primary directions.

The low-rank examples give a complementary interpretation of the secondary module. Fixing the primary invariants can leave a finite set of gauge-inequivalent configurations distinguished only by secondary data. A potential constructed entirely from the primaries will then localize on several dynamically accessible points of this finite fiber. In this way the Hironaka decomposition provides a natural mechanism to generate multiple ground states in matrix quantum mechanics. Sources coupled to a secondary can select among these points, fluctuations across different sheets can be distinguished from ordinary perturbative fluctuations within one sheet, and finite-action instantons can produce exponentially suppressed mixing between the corresponding semiclassical states. These examples show that secondary invariants can label dynamical semiclassical sectors that are invisible to perturbation theory about a single sheet.

This interpretation is dynamically realizable but not universal. The existence of a secondary invariant does not by itself guarantee distinct classical minima, a large barrier, an anomalous variance or an instanton connecting different fiber points. Nor need an interacting Hamiltonian preserve the individual primary towers in the Hironaka decomposition. An important problem is therefore to identify signatures of the finite fiber that are independent of the chosen primary coordinates. It would also be valuable to determine the optimized onset $\Delta_{N,d}$, to test
ballistic growth directly in matrix quantum mechanics, and to study whether the same sector structure is realized by physically natural Hamiltonians, in particular those with a commutator-squared Yang--Mills potential.

The heavy sector built from $O(N^2)$ elementary fields is a particularly natural setting in which to pursue these questions. In this regime the proliferation of trace species competes directly with the $O(N^2)$ continuous capacity of the invariant space, while the finite secondary module contains exponentially many elements. This scaling is
suggestive of the entropy carried by matrix degrees of freedom and may be relevant to the reorganization of the Hilbert space associated with black-hole physics.

The broader conclusion is that the Hironaka decomposition is more than an efficient method for solving finite-$N$ trace relations. It separates the exact invariant algebra into a freely generated coordinate algebra and the finite algebraic information required to complete it. The theorem proved here ensures that the freely generated coordinates may be chosen to be single traces, preserving their usual interpretation as perturbative single-particle observables, while the secondary module records the information absent from the large-$N$ Fock description. Overcrowding quantifies when this finite completion becomes globally unavoidable. The Hironaka decomposition therefore provides a natural language for finite-$N$ collective field theory: it retains the perturbative single-trace description where it is valid while exposing the additional algebraic structure needed to describe the exact finite-$N$ Hilbert space.

\section*{Acknowledgments}

The work of R.d.M.K. and A.R. is supported by a start-up research fund of Huzhou Normal University, a Zhejiang Province talent award and a Changjiang Scholar award.

\appendix

\section{Technical ingredients in the single-trace construction}
\label{app:technical-single-trace}

\subsection{The criterion for choosing primary invariants}
\label{app:hsop-criterion}

There is a simple criterion that tells us when a set of homogeneous invariants can be used as primary invariants.  Let
\begin{equation}
R=\bigoplus_{m\geq 0}R_m
\end{equation}
be the graded ring of invariant polynomials, with $R_0=\mathbb C$, and let $h=\dim R$ be the number of independent continuous gauge-invariant directions.  $h$ is the number of independent coordinates needed to describe the continuous part of the space of invariant data.

Suppose we have chosen $h$ homogeneous invariants $p_1,\ldots,p_h$. These can be taken as primary invariants if they detect all continuous directions. Equivalently, once the $p_a$ are fixed, there should be no remaining continuous freedom in the invariant data.  There may still be a finite number of algebraic possibilities. This finite data is encoded by the secondary invariants.

To test a proposed set of primary invariants, look at the common zero locus of the $p_a$.  Assume that we have managed to argue that
\begin{equation}
p_1=p_2=\cdots=p_h=0
\end{equation}
forces the invariant data to sit at the invariant origin, namely the point where all positive-degree invariants vanish.  In algebraic terms this means that the ideal generated by the $p_a$ has the same zero locus as the ideal of all positive-degree invariants\footnote{If $I\subset R$ is an ideal, its radical is defined by
\begin{equation}
\sqrt I=\left\{f\in R ;\middle|; f^k\in I\text{ for some positive integer } k\right\}.
\end{equation}
Thus $f$ belongs to $\sqrt I$ if some power of $f$ belongs to $I$.}
\begin{equation}
\sqrt{(p_1,\ldots,p_h)}=R_+, \qquad R_+=\bigoplus_{m>0}R_m .\label{eq:radical-origin}
\end{equation}
Thus the equations $p_a=0$ leave no nontrivial continuous direction through the origin.  Since the $p_a$ are homogeneous, the same conclusion holds after fixing generic values of the $p_a$: their values remove all continuous freedom, up to a finite ambiguity.

We can make this more explicit.  Consider the quotient $R/(p_1,\ldots,p_h)$. In this quotient we set all proposed primary invariants to zero.  By the assumption above, this kills all continuous directions in the invariant space.  Therefore only a finite-dimensional vector space of residual invariants remains.  Choose homogeneous representatives $q_1,\ldots,q_s$ for a basis of this quotient.  These representatives are the finite set of possible ``remainders'' left after reducing by the $p_a$.

We claim that every invariant in $R$ can be written in the form
\begin{equation}
F=\sum_{\alpha=1}^s f_\alpha(p_1,\ldots,p_h)\, q_\alpha ,\label{eq:module-expansion}
\end{equation}
where the $f_\alpha$ are polynomials.  This is the algebraic version of the statement that the $p_a$ describe the continuous directions, while the $q_\alpha$ describe the finite residual data.

The proof is by degree.  Take a homogeneous invariant $F$.  Its image in the quotient $R/(p_1,\ldots,p_h)$ can be expressed as a linear combination of the chosen representatives $q_\alpha$.  Subtracting this combination from $F$, the remainder lies in the ideal generated by the $p_a$:
\begin{equation}
F-\sum_\alpha c_\alpha q_\alpha=\sum_{a=1}^h p_a F_a .
\end{equation}
Since each $p_a$ has positive degree, the coefficients $F_a$ have lower degree than $F$.  Repeating the same argument for the $F_a$, and continuing by induction on the degree, gives the expansion \eqref{eq:module-expansion}.  Thus $R$ is a finite module over the polynomial ring generated by the $p_a$
\begin{equation}
\mathbb C[p_1,\ldots,p_h]\subset R .
\end{equation}

It remains to check that the $p_a$ are algebraically independent.  The ring $R$ has dimension $h$, and we have just shown that $R$ is finite over $\mathbb C[p_1,\ldots,p_h]$.  A finite extension does not change the number of continuous directions.  Hence
\begin{equation}
\dim \mathbb C[p_1,\ldots,p_h]=\dim R=h .
\end{equation}
But a ring generated by $h$ elements has dimension $h$ only when those $h$ generators obey no algebraic relation among themselves.  Therefore $p_1,\ldots,p_h$ are algebraically independent.

We have shown that the $p_a$ are algebraically independent and that $R$ is finite over the polynomial ring $\mathbb C[p_1,\ldots,p_h]$.  This is precisely the statement that $p_1,\ldots,p_h$ form a homogeneous system of parameters.  In the language used in the main text, they may be chosen as the primary invariants.  The remaining finite set of representatives $q_\alpha$ are the secondary invariants.

\subsection{Proof of the common-degree lemma}\label{app:common-degree}

Represent the finitely many invariant-data assignments $q_j$ by matrix tuples $A^{(j)}=(A_1^{(j)},\ldots,A_d^{(j)})$ in the complexified configuration space.  These representatives need not be Hermitian; the complexification is used only to prove a statement about the same invariant ring that governs the Hermitian model.  Because $q_j$ is not the invariant origin, choose a word $w_j$ of length $\ell_j$ such that
\begin{equation}
        \Tr\left(w_j(A^{(j)})\right)\neq0.\label{eq:detecting-word}
\end{equation}
Set $W_j=w_j(A^{(j)})$.  The matrix $W_j$ is not nilpotent.  Let
\begin{equation}
        L=\operatorname{lcm}(\ell_1,\ldots,\ell_s),\qquad  B_j=W_j^{L/\ell_j}.
\end{equation}
Every $B_j$ remains nonnilpotent, but its trace may vanish by cancellation between eigenvalues.  We now choose one further power that avoids these cancellations simultaneously.

Let $\rho_j>0$ be the largest modulus of an eigenvalue of $B_j$, and write the eigenvalues of modulus $\rho_j$ as $\rho_j e^{i\theta_{j,a}}$.  The finite collection of phases defines a point on a compact torus.  Recurrence of powers on a compact torus gives a sequence $q_\nu\to\infty$ such that
\begin{equation}
        e^{iq_\nu\theta_{j,a}}\longrightarrow1
\end{equation}
simultaneously for every $j$ and $a$.  Along this sequence,
\begin{equation}
        \rho_j^{-q_\nu}\Tr(B_j^{q_\nu}) =\sum_{|\lambda|=\rho_j} \left(\frac{\lambda}{\rho_j}\right)^{q_\nu} +o(1)
        \longrightarrow m_j>0,  \label{eq:dominant-eigenvalue-alignment}
\end{equation}
where $m_j$ is the number of eigenvalues of maximal modulus, counted with multiplicity.  Hence one sufficiently large common value $q=q_\nu$ obeys
\begin{equation}
        \Tr(B_j^q)\neq0 \qquad\text{for every }j.
\end{equation}
Finally define
\begin{equation}
        t_j =\Tr\!\left(w_j^{qL/\ell_j}\right).
\end{equation}
All $t_j$ have the same degree $M=qL$, and
\begin{equation}
        t_j(q_j)=\Tr(B_j^q)\neq0.
\end{equation}
This proves Lemma~\ref{lem:common-degree}.

\end{document}